\documentclass[final]{siamltex}
\usepackage{subfigure}
\usepackage{graphicx}
\usepackage{algpseudocode, algorithm}
\usepackage{multirow}
\usepackage{amsmath}
\usepackage{amssymb}
\usepackage{paralist}
\usepackage{cite}

\usepackage{color}

\newcommand{\real}{\mathbb{R}}

\usepackage{url}

\title{A Method Based on Total Variation for Network Modularity Optimization using the MBO Scheme\thanks{This work was supported by UC Lab Fees Research grant 12-LR-236660, ONR grant N000141210838, ONR grant N000141210040, AFOSR MURI grant FA9550-10-1-0569, NSF grant DMS-1109805. M.A.P. was supported by a research award (\#220020177) from the James S. McDonnell Foundation, the EPSRC (EP/J001759/1), and the FET-Proactive project PLEXMATH (FP7-ICT-2011-8; grant \#317614) funded by the European Commission.}}

\author{Huiyi Hu\thanks{Department of Mathematics, University of California, Los Angeles. Los Angeles, CA, USA. ({\tt huiyihu@math.ucla.edu, bertozzi@math.ucla.edu})} \and Thomas Laurent\thanks{Department of Mathematics, University of California, Riverside. Riverside, CA, USA. ({\tt laurent@math.ucr.edu})} \and Mason A. Porter\thanks{Oxford Centre for Industrial and Applied Mathematics, Mathematical Institute; and CABDyN Complexity Centre, University of Oxford, Oxford, UK. ({\tt porterm@maths.ox.ac.uk})} \and Andrea L. Bertozzi$^\dagger$}

\begin{document}

\maketitle

%%%%%%%%%%%%%

\begin{abstract}
The study of network structure is pervasive in sociology, biology, computer science, and many other disciplines. One of the most important areas of network science is the algorithmic detection of cohesive groups  of nodes called ``communities''.  One popular approach to find communities is to maximize a quality function known as {\em modularity} to achieve some sort of optimal clustering of nodes. In this paper, we interpret the modularity function from a novel perspective: we reformulate modularity optimization as a minimization problem of an energy functional that consists of a total variation term and an $\ell_2$ balance term. By employing numerical techniques from image processing and $\ell_1$ compressive sensing---such as convex splitting and the Merriman-Bence-Osher (MBO) scheme---we develop a variational algorithm for the minimization problem. We present our computational results using both synthetic benchmark networks and real data.
\end{abstract}



\pagestyle{myheadings}
\thispagestyle{plain}
\markboth{H. Hu, T. Laurent, M. A. Porter\and A. L. Bertozzi}{A Method Based on Total Variation for Network Modularity Optimization}

%%%%%%%%%%%%%%

\section{Introduction}

Networks provide a useful representation for the investigation of complex systems, and they have accordingly attracted considerable attention in sociology, biology, computer science, and many other disciplines \cite{NewmanPhys08,newman2010}.  Most of the networks that people study are graphs, which consist of nodes (i.e., vertices) to represent the elementary units of a system and edges to represent pairwise connections or interactions between the nodes.  

Using networks makes it possible to examine intermediate-scale structure in complex systems.  
Most investigations of intermediate-scale structures have focused on {\em community structure}, in which one decomposes a network into (possibly overlapping) cohesive groups of nodes called {\em communities} \cite{MasonNotice}.\footnote{Other important intermediate-scale structures to investigate include core-periphery structure \cite{puck} and block models \cite{doreian}.} There is a higher density of connections within communities than between them.

In some applications, communities have been related to functional units in networks \cite{MasonNotice}.  For example, a community might be closely related to a functional module in a biological system \cite{anna} or a group of friends in a social system \cite{MasonFacebook}. Because community structure in real networks can be very insightful \cite{MasonNotice, GN02,Fortunato09,newman2010}, it is useful to study algorithmic methods to detect communities.   Such efforts have been useful in studies of the social organization in friendship networks \cite{MasonFacebook}, legislation cosponsorships in the United States Congress \cite{MasonCongress}, functional modules in biology networks \cite{GuimeraBiology,anna}, and many other situations.

To perform community detection, one needs a quantitative definition for what constitutes a community, though this relies on the goal and application one has in mind. Perhaps the most popular approach is to optimize a quality function known as {\em modularity} \cite{newman03,GN04,Newmanspectral}, and numerous computational heuristics have been developed for optimizing modularity \cite{MasonNotice,Fortunato09}. The modularity of a network partition measures the fraction of total edge weight within communities versus what one might expect if edges were placed randomly according to some null model.  We give a precise definition of modularity in equation (\ref{modularity}) in Section \ref{modreview}. Modularity gives one definition of the ``quality'' of a partition, and maximizing modularity is supposed to yield a reasonable partitioning of a network into disjoint communities. 

Community detection is related to 
%the network clustering problem known as 
{\em graph partitioning}, which has been applied to problems in numerous areas (such as data clustering) \cite{tutorialSC, ShiMalik,WeissNg}. %{\bf map: (a) do we want refs here? (b) do we want to list other areas [with refs] too?} {\bf [Huiyi: is it good enough to cite \cite{tutorialSC, ShiMalik}? Perhaps no need to list other areas?]} 
%Traditional graph partitioning seeks to divide a network's nodes into clusters with predefined sizes \blue{(the size is not predefined in spectral clustering, NMF, etc...)} such that the total edge weight between clusters is minimized \cite{Newmanspectral, newman2010} \blue{(these two refs belong to the community detection literature rather than the clustering literature)}.  
%{\bf [Thomas says it is not true to claim graph partitioning requires predefined sizes. so I removed this sentence. Also, I cited refs from the graph partitioning literature, instead of community detection literature.]}
In graph partitioning, a network is divided into disjoint sets of nodes. Graph partitioning usually requires the number of clusters to be specified to avoid trivial solutions, whereas modularity optimization does not require one to specify the number of clusters \cite{MasonNotice}. This is a desirable feature for applications such as social and biological networks.

Because modularity optimization is an NP-hard problem \cite{NPhard}, efficient algorithms are necessary to find good locally optimal network partitions with reasonable computational costs.  Numerous methods have been proposed \cite{MasonNotice,Fortunato09}.  These include greedy algorithms \cite{newman04,Clauset04}, extremal optimization \cite{EO1,EO2}, simulated annealing \cite{SA1,SA2}, spectral methods (which use eigenvectors of a modularity matrix) \cite{Newmanspectral,tripart}, and more.  The locally greedy algorithm by Blondel et al. \cite{Blondel} is arguably the most popular computational heuristic; it is a very fast algorithm, and it also yields high modularity values \cite{Fortunato09, LFRcomp}.

In this paper, we interpret modularity optimization (using the Newman-Girvan null model \cite{GN04,newman2010}) from a novel perspective. Inspired by the connection between graph cuts and the total variation (TV) of a graph partition, we reformulate the problem of modularity optimization as a minimization of an energy functional that consists of a graph cut (i.e., TV)  term and an $\ell_2$ balance term. By employing numerical techniques from image processing and $\ell_1$ compressive sensing---such as convex splitting and the Merriman-Bence-Osher (MBO) scheme \cite{MBO}---we propose a variational algorithm to perform the minimization on the new formula. We apply this method to both synthetic benchmark networks and real data sets, and we achieve performance that is competitive with the state-of-the-art modularity optimization algorithms. 
%Our method is particularly fast for large networks that contain a small number of clusters.  {\bf map: actually, we don't want to say this because for most networks it is likely true that one should expect a large number of small clusters rather than a small number of large ones; so while that should be mentioned in the body of the paper, I would much rather just state here that we are competitive with the state-of-the-art, as this statement would make many referees --- including me --- start thinking in directions we don't want}

%%% Especially for very large networks with smaller numbers of clusters, we obtain significant speed-up over prior methods (by a factor around 40 for a 70,000-node network).

The rest of this paper is organized as follows. In Section \ref{meth}, we review the definition of the modularity function, and we then derive an equivalent formula of modularity optimization as a minimization problem of an energy functional that consists of a total variation term and an $\ell_2$ balance term. In Section \ref{section algorithm}, we explain the MBO scheme and convex splitting, which are numerical schemes that we employ to solve the minimization problem that we proposed in Section \ref{meth}. In Section \ref{numres}, we test our algorithms on several benchmark and real-world networks. 
%including LFR benchmark, MINIST digit images data, and network coauthorship. 
We then review the similarity measure known as the
%We then discuss the network diagnostics of the results used to analyze the partitions and evaluating the performance. In particular, a similarity measure known as the 
{\em normalized mutual information} (NMI) and use it to compare network partitions with ground-truth partitions.  We also evaluate the speed of our method, which we compare to classic spectral clustering \cite{ShiMalik, tutorialSC}, modularity-based spectral partitioning \cite{Newmanspectral, tripart}, and the GenLouvain code \cite{Netwiki} (which is an implementation of a Louvain-like algorithm \cite{Blondel}). In Section \ref{conc}, we summarize and discuss our results.

%%%%%%%%%%%

\section{Method}\label{meth}

Consider an $N$-node network, which we can represent as a weighted graph $(G, E)$ with a node set $G = \{n_1,n_2,\ldots,n_N\}$ and an edge set $E = \{w_{ij}\}_{i,j=1}^N$. The quantity $w_{ij}$ indicates the closeness (or similarity) of the tie between nodes $n_i$ and $n_j$, and the array of all $w_{ij}$ values forms the graph's \emph{adjacency matrix} $\textbf{W}=[w_{ij}]$. In this work, we only consider undirected networks, so $w_{ij}=w_{ji}$. 

%%%%%%

\subsection{Review of the Modularity Function} \label{modreview}
The modularity of a graph partition measures the fraction of total edge weight within each community minus the edge weight that would be expected if edges were placed randomly using some null model \cite{MasonNotice}. The most common null model is the Newman-Girvan (NG) model \cite{GN04}, which assigns the expected edge weight between $n_i$ and $n_j$ to be $\frac{k_i k_j}{2m}$, where $k_i=\sum_{s=1}^N w_{is}$ is the strength (i.e., weighted degree) of $n_i$ and $2m=\sum_{i=1}^N k_i$ the total volume (i.e., total edge weight) of the graph $({G,E})$.  When a network is unweighted, then $k_i$ is the degree of node $i$. An advantage of the NG null model is that it preserves the expected strength distribution of the network.  
%We will use the NG null model (with an additional resolution parameter) throughout this paper.
% {\bf map: I removed this sentence so that I wouldn't have to discuss the resolution parameter before the formula below}

A \emph{partition}  $g=\{g_i\}_{i=1}^N$ of the graph $(G,E)$ consists of a set of disjoint subsets of the node set $G$ whose union is the entire set $G$.  The quantity $g_i \in\{1,2,\ldots,\hat{n}\}$ is the community assignment of $n_i$, where there are $\hat{n}$ communities ($\hat{n}\leq N$). The \emph{modularity} of the partition $g$ is defined as
\begin{equation}
	Q(g) =  \frac{1}{2m}\sum_{i,j=1}^N \left(w_{ij}-\gamma \frac{k_i k_j}{2m}\right)\delta(g_i,g_j)\,, \label{modularity}
\end{equation}
where $\gamma$ is a resolution parameter \cite{resolution}. The term $\delta(g_i,g_j)=1$ if $g_i=g_j$ and $\delta(g_i,g_j)=0$ otherwise. The resolution parameter can change the scale at which a network is clustered \cite{MasonNotice,Fortunato09}. A network breaks into more communities as one increases $\gamma$. 

By maximizing modularity, one expects to obtain a reasonable partioning of a network. However, this maximization problem is NP hard \cite{NPhard}, so considerable effort has been put into the development of computational heuristics to obtain network partitions with high values of $Q$.

%{\bf map: note: above I have written $\hat{n}$ communities, though later on to set up the current formulation we want to say 'at most' $\hat{n}$.  I think this abuse of notation is ok and probably easier than having two separate notations regarding that, but I wanted to at least broach the subject.}

%%%%%%%%%

\subsection{Reformulation of Modularity Optimization}

In this subsection, we reformulate the problem of modularity optimization by
%derive a new expression for $Q$ that reformulates the problem of modularity optimization problem. We 
deriving a new expression for $Q$ that bridges the network-science and compressive-sensing communities.  This formula makes it possible to use techniques from the latter to tackle the modularity-optimization problem with low computational cost.

We start by defining the {\em total variation} (TV), weighted $\ell_2$-norm, and weighted mean of a function $f: G \rightarrow \real$:
\begin{align}
	|f|_{TV} &:=   \frac{1}{2} \sum_{i,j=1}^Nw_{ij} \left|f_i-f_j\right|\,, \notag \\
	\quad \|f\|_{\ell_2}^2 &:=    \sum_{i=1}^N  k_i  \left|f_i\right|^2\,,  \notag \\
	\mathrm{mean}(f) &:= \frac{1}{2m}\sum_{i=1}^N k_i f_i \, ,\label{defTV}
\end{align}
where $f_i=f(n_i)$. The quantity $\frac{1}{2} \sum_{i,j=1}^Nw_{ij} |f_i-f_j| $ is called the total variation because it enjoys many properties of the classical total variation $\int |\nabla f|$ of a function $f: \real^n \to \real$ \cite{fan}.
For a vector-valued function  $f=(f^{(1)}, \ldots, f^{(\hat n)})$: $G\rightarrow \real^{\hat{n}}$, we define 
\begin{align}
	|f|_{TV} &:=  \sum_{l=1}^{\hat n} |f^{(l)}|_{TV}\,, \notag \\
	\|f\|_{\ell_2}^2 &:=    \sum_{l=1}^{\hat n}  \|f^{(l)}\|_{\ell_2}^2\,,   
  \label{defTV2}
\end{align}
and $\mathrm{mean}(f):= \left(\mathrm{mean}(f^{(1)}), \ldots, \mathrm{mean}(f^{(\hat n)})\right)$.

%{\bf map: not fair to have an ambiguous, unexplained statement of "less" computational cost without having a thorough comparison of the best method, not just our implementation of it; it is fair to state that it's low, but we need to be really careful with our precise phrasing of such things}

Given the partition $g=\{g_i\}_{i=1}^N$ defined in Section \ref{modreview}, let $A_l = \{n_i\in G, g_i = l\}$, where $l\in\{1,2,\ldots,\hat{n}\}$ ($\hat{n} \leq N$). Thus, $G=\cup_{l=1}^{\hat{n}}A_l$ is a partition of the network $(G,E)$ into disjoint communities. Note that every $A_l$ is allowed to be empty, so $g$ is a partition into {\it at most} $\hat n$ communities.   
 Let $f^{(l)}: G \to \{0,1\}$ be the indicator function of community $l$; in other words, $f^{(l)}_i$ equals one if $g_i=l$, and it equals zero otherwise.  The function $f=(f^{(1)}, \ldots, f^{(\hat{n})})$ is then called the \emph{partition function} (associated with $g$).
 Because each set $A_l$ is disjoint from all of the others, it is guaranteed that only a single entry of $f_i$ equals one for any node $i$. Therefore, $f: G \to V^{\hat{n}} \subset \real^{\hat n}$, where  $V^{\hat{n}}$
\begin{equation*} 
 	V^{\hat{n}}:=\{(1,0,\ldots,0),(0,1,0,\ldots,0),\ldots,(0,\ldots,0,1)\}=\{\vec{e}_l\}_{l=1}^{\hat{n}}  
\end{equation*}
is the standard basis of $\mathbb{R}^{\hat{n}}$. 

The key observation that bridges the network-science and compressive-sensing communities is the following: 

\

\begin{theorem} Maximizing the modularity functional $Q$ over all partitions that have at most $\hat{n}$ communities is equivalent to minimizing
\begin{align}
	|f|_{TV}-\gamma \|f-\mathrm{mean}(f)\|^2_{\ell_2}  \label{minimization}
\end{align}
over all  functions $f: G \to V^{\hat{n}}$.
\end{theorem}

\

\begin{proof}
In the language of graph partitioning, $\mathrm{vol}(A_l) :=\sum_{n_i\in A_l}k_i$ denotes the volume of the set $A_l$, and $\mathrm{Cut}(A_l,A_l^c):=\sum_{n_i\in A_l,n_j\in A_l^c}w_{ij}$ is the graph cut of $A_l$ and $A_l^c$. Therefore,
\begin{align}
	 Q(g) &=\frac{1}{2m}\left[ \big(2m-  \sum_{g_i\neq g_j} w_{ij}\big)-\frac{\gamma}{2m}\sum_{l=1}^{\hat{n}}\left(\sum_{n_i\in A_l, n_j\in A_l}k_i k_j \right) \right] \nonumber  \\
	&=1-\frac{1}{2m}\left( \sum_{l=1}^{\hat{n}} \mathrm{Cut}(A_l,A_l^c) +  \frac{\gamma}{2m}\sum_{l=1}^{\hat{n}}\mathrm{vol}A_l^2\right) \nonumber \\
	&=1-\gamma -\frac{1}{2m}\left(\sum_{l=1}^{\hat{n}} \mathrm{Cut}(A_l,A_l^c) -\frac{\gamma}{2m} \Big(\sum_{l=1}^{\hat{n}}\mathrm{vol}A_l \cdot \mathrm{vol}A_l^c \Big)\right)\,,\label{rewrite}
\end{align}
where the sum $\sum_{g_i\neq g_j} w_{ij}$ includes both $w_{ij}$ and $w_{ji}$. %[Thomas thinks the sum is ambiguous and need this clarification.]
Note that if  $\chi_A: G \to \{0,1\}$ is the  indicator function of a subset $A\subset G$, then $\left|\chi_A\right|_{TV}=\mathrm{Cut}\left(A,A^c\right)$ and 
\begin{align*}
	\left\|\chi_A-\mathrm{mean}(\chi_A)\right\|^2_{\ell_2}&=\sum_{i=1}^Nk_i\left|\chi_A(n_i)-\frac{\mathrm{vol}(A)}{2m}\right|^2\nonumber\\ 
	&= \mathrm{vol}(A) \left(1-\frac{\mathrm{vol}(A)}{2m}\right)^2 + \mathrm{vol}\left(A^c\right) \left(\frac{\mathrm{vol}\left(A\right)}{2m}\right)^2 \notag \\
	&=\frac{\mathrm{vol}(A) \cdot \mathrm{vol}\left(A^c\right)}{2m}\,.
\end{align*}
Because  $f^{(l)}=\chi_{A_l}$ is the indicator function of $A_l$, it follows that
\begin{align}
	\left|f\right|_{TV}-\gamma \|f-\mathrm{mean}(f)\|^2_{\ell_2} &=\sum_{l=1}^{\hat{n}} \left\{ |f^{(l)}|_{TV} -\gamma \|f^{(l)}-\mathrm{mean}(f^{(l)})\|^2_{\ell_2}\right\}  \nonumber \\
	&=\sum_{l=1}^{\hat{n}}\left\{ \mathrm{Cut}(A_l,A_l^c)-\gamma \frac{\mathrm{vol}(A_l) \cdot \mathrm{vol}(A_l^c)}{2m}\right\}\,. \label{L2}
\end{align}

Combining \eqref{rewrite} and \eqref{L2}, we conclude that maximizing $Q$ is equivalent to minimizing \eqref{minimization}. 
 \end{proof}

\

With the above argument, we have reformulated the problem of modularity maximization as the minimization problem (\ref{minimization}), which corresponds to minimizing the total variation (TV) of the function $f$ along with a balance term.  This yields a novel view of modularity optimization that uses the perspective of compressive sensing (see the references in \cite{CS-MRI}). In the context of compressive sensing, one seeks a solution of function $f$ that is compressible under the transform of a linear operator $\Phi$. That is, we want $\Phi f$ to be well-approximated by sparse functions. (A function is considered to be ``sparse'' when it is equal to or approximately equal to zero on a ``large" portion of the whole domain.) Minimizing $\|\Phi f\|_{\ell^1}$ promotes sparsity in $\Phi f$. When $\Phi$ is the gradient operator (on a continuous domain) or the finite-differencing operator (on a discrete domain) $\nabla$, then the object $\|\Phi f\|_{\ell^1}=\|\nabla f\|_{\ell^1}$ becomes the total variation $|f|_{TV}$ \cite{Image-via-TV, CS-MRI}. The minimization of TV is also common in image processing and computer vision \cite{ROF, Image-via-TV, CS-MRI,ChanVese}. 
 
% {\bf map: we need a quick parenthetical comment to help readers from the network-science side; I see the terminology "sparse function": does this mean that the support is 'small' in some sense or is it defined in a different way?  anyway, please add this clarificiation}

%\red{ As discussed for example in Ref. \cite{Newmanspectral}, the graph cut problem requires a constraint on group sizes to avoid trivial partitions. Rather, modularity is a modified quality function related to the cut which compares the network with a probability null model and need no such constraint. } {\bf [Huiyi: you told me to add something mentioning about Newman's 2006 paper and this is what I wrote after I went through his paper carefully. But actually I don't think Newman's paper makes any claims similar to what we present in the following. So perhaps we can remove this paragraph? In the following paragraph, Thomas rephrased our contribution explicitly as "gives a geometric interpretation of modularity optimization", which definitely didn't show up in Newman's paper. Correct me if I understood it wrong.]}
 
The expression in equation (\ref{rewrite}) is interesting because its geometric interpretation of modularity optimization contrasts with existing interpretations (e.g., probabilistic ones or in terms of the Potts model from statistical physics \cite{Newmanspectral,MasonNotice}). For example, we see from \eqref{rewrite} that finding the  bipartition of the graph $G=A\cup A^c$ with maximal modularity is equivalent to minimizing 
\begin{equation*}
 	\mathrm{Cut}(A,A^c)-\frac{\gamma}{2m} \mathrm{vol}(A) \cdot \mathrm{vol}\left(A^c\right)\,.
\end{equation*}	
Note that the term  $\mathrm{vol}(A) \cdot \mathrm{vol}\left(A^c\right)$ is maximal when $\mathrm{vol}(A) = \mathrm{vol}\left(A^c\right)=m$. Therefore, the second term is a {\it balance term} that favors a partition of the graph into two groups of roughly equal size. In contrast, the first term favors a partition of the graph in which few links are severed.  This is reminiscent of the {\em Balance Cut} problem in which the objective is to minimize the ratio 
 \begin{equation}
 	\frac{\mathrm{Cut}\left(A,A^c\right)}{\mathrm{vol}(A) \cdot \mathrm{vol}\left(A^c\right)}\,. \label{balancecut}
\end{equation}	
 In recent papers Refs.~\cite{SB09,pro:HeinBuhler10OneSpec, pro:HeinSetzer11TightCheeger,
art:BressonTaiChanSzlam12TransLearn, pro:Rang-Hein-constrained,BLUV12 }, various TV-based algorithms were proposed to minimize ratios similar to \eqref{balancecut}.

%%%%%%%%================Algorithm============%%%%%%%%%%%%%%%%%%%%%

\section{Algorithm}\label{section algorithm}
Directly optimizing \eqref{minimization} over all partition functions $f: G \to V^{\hat{n}}$
 is difficult due to the discrete solution space. Continuous relaxation is thus needed to simplify the optimization problem.

%%%%%

\subsection{Ginzburg-Landau Relaxation of the Discrete Problem}

Let $X^p$ 
\begin{equation*}
	X^p=\{f \; | \;f:G \to V^{\hat{n}}  \}
\end{equation*}
denote the  space of  partition functions. Minimizing \eqref{minimization} over  $X^p$ is equivalent  to minimizing 
\begin{equation}
	H(f)=\begin{cases}
|f|_{TV}-\gamma \|f-\mathrm{mean}(f)\|^2_{\ell_2}\,, & \text{if } f \in X^p \\
+\infty\,, & \text{otherwise}
\end{cases}
\end{equation}
over all $f: G \to \real^{\hat n}$.  

The Ginzburg-Landau (GL) functional has been used as an alternative for the TV term in image processing (see the references in Ref.~\cite{flenner}) due to its $\Gamma$-convergence to the TV of the characteristic functions in Euclidean space \cite{GL-TV}. Reference~\cite{flenner} developed a graph version of the GL functional and used it for graph-based high-dimensional data segmentation problems. The authors of Ref.~\cite{MBO-multiGL} generalized the two-phase graphical GL functional to a multi-phase one. 

For a graph $(G,E)$, the (combinatorial) {\em graph Laplacian} \cite{fan} is defined as
\begin{equation}	
	\textbf{L} = \textbf{D}-\textbf{W}\,,
\end{equation}
where $\textbf{D}$ is a diagonal matrix with nodes of strength $\{k_i\}_{i=1}^N$ on the diagonal and $\textbf{W}$ is the weighted adjacency matrix. The operator $\textbf{L}$ is linear on $\{z | z:G\to \real\}$, and satisfies:
%the following quadratic equality holds: 
\begin{equation*}	 
	 \langle z,\textbf{L}z \rangle = \frac{1}{2}\sum_{i,j} w_{ij}(z_i-z_j)^2\,,
\end{equation*}	 
where $z_i=z(n_i)$ and $i\in\{1,2,\ldots,N\}$. 

Following the idea in Refs.~\cite{flenner, MBO-multiGL}, we define the {\em Ginzburg-Landau relaxation} of $H$ as follows:
\begin{align}
	H_\epsilon(f)=\frac{1}{2}\sum_{l=1}^{\hat{n}}\langle f^{(l)},\textbf{L}f^{(l)} \rangle +\frac{1}{\epsilon^2}\sum_{i=1}^N W_{\mathrm{multi}}(f_i)-\gamma \|f-\mathrm{mean}(f)\|^2_{\ell_2}\, ,
\label{relaxation}
\end{align}
where $\epsilon >0$.
In equation (\ref{relaxation}), $W_{\mathrm{multi}}: \real^{\hat{n}} \to \real$ is a multi-well potential (see Ref. \cite{MBO-multiGL}) with equal-depth wells. The minima of $W_{\mathrm{multi}}$ are spaced equidistantly, take the value $0$, and correspond to the points of $V^{\hat{n}}$. The specific formula for $W_{\mathrm{multi}}$ does not matter for the present paper, because we will discard it when we implement the MBO scheme. Note that the purpose of this multi-well term is to force $f_i$ to go to one of the minima, so that one obtains an approximate phase separation.

Our next theorem states that modularity optimization with an upper bound on the number of communities is well-approximated (in terms of $\Gamma$-convergence) by minimizing $H_\epsilon$ over all $f: G \to \real^{\hat{n}}$.  Therefore, the {\em discrete} modularity optimization problem \eqref{minimization} can be approximated by a {\em continuous} optimization problem. We give the mathematical definition and relevant proofs of $\Gamma$-convergence in the Appendix.

\

\begin{theorem}[$\Gamma$--convergence of $H_\epsilon$ towards $H$]\label{theorem relaxation}
The functional $H_\epsilon$ $\Gamma$-converges to $H$ on the space $X=\{f \; | \;f:G \to \real^{\hat{n}}  \}$.
\end{theorem}

\

\begin{proof}
As shown in Theorem \ref{gamma convergence} (in the Appendix), $H_\epsilon + \gamma \|f-\mathrm{mean}(f)\|^2_{\ell_2}$ $\Gamma$-converges to $H+\gamma \|f-\mathrm{mean}(f)\|^2_{\ell_2}$ on $X$. Because $\gamma \|f-\mathrm{mean}(f)\|^2_{\ell_2}$
is continuous on the metric space $X$, it is straightforward to check that $H_\epsilon$ $\Gamma$-converges to $H$ according to the definition of $\Gamma$-convergence.
\end{proof}

\

By definition of $\Gamma$-convergence, Theorem \ref{theorem relaxation} directly implies the following:

\

\begin{corollary}
Let $f^\epsilon$ be the global minimizer of $H_\epsilon$. Any convergent subsequence of $f_{\epsilon}$ then converges to a global maximizer of the modularity $Q$ with at most $\hat{n}$ communities.  
\end{corollary}

%%%%%%%

\subsection{MBO Scheme, Convex splitting, and Spectral Approximation} %\label{sectionMBOscheme} 

In this subsection, we use techniques from the compressive-sensing and image-processing literatures to develop an efficient algorithm that (approximately) optimizes $H_\epsilon$.

%
%Motivated by the $\Gamma$-convergence of $GL^{\mathrm{multi}}_{\epsilon}(f)$ to the graph TV, we replace $|f|_{TV}$ by $GL^{\mathrm{multi}}_{\epsilon}(f)$ in (\ref{minimization}) to obtain the energy 
%\begin{equation}
%	H_{\epsilon}(f) = GL^{\mathrm{multi}}_{\epsilon}(f)-\gamma \|f-\mathrm{mean}(f)\|^2_{L_2}\,,
%\end{equation}	
%which contains a balance term in addition to the GL term.

In Ref.~\cite{MBO}, an efficient algorithm (which is now called the \emph{MBO scheme}) was proposed to approximate the gradient descent of the GL functional using threshold dynamics. See Refs.~\cite{MBOconverge,MBOconverge2,Selim} for discussions of the convergence of the MBO scheme. Inspired by the MBO scheme, the authors of Ref. \cite{SelimTsai} developed a method using a PDE framework to minimize the piecewise-constant Mumford-Shah functional (introduced in Ref. \cite{Mumford}) for image segmentation. Their algorithm was motivated by the Chan-Vese level-set method \cite{ChanVese} for minimizing certain variants of the Mumford-Shah functional. Note that the Chan-Vese method is related to our reformulation of modularity, because it uses the TV as a regularizer along with $\ell_2$ based fitting terms. The authors of Refs.~\cite{MBO-tiana, MBO-multiGL} applied the MBO scheme to graph-based problems.

The gradient-descent equation of \eqref{relaxation} is
\begin{align}
	\frac{\partial f}{\partial t}=- (\mathbf{L}f^{(1)},\ldots,\mathbf{L}f^{(\hat{n})}) - \frac{1}{\epsilon^2} \nabla W_{\mathrm{multi}}(f) +\frac{\delta}{\delta f} \left(\gamma \|f-\mathrm{mean}(f)\|^2_{\ell_2}\right)\,,\label{GradientDescent}
\end{align}
where $\nabla W_{\mathrm{multi}}(f) : G \to \real^{\hat{n}}$ is the composition of the functions $\nabla W_{\mathrm{multi}}$ and $f$.  Thus, one can follow the idea of the original MBO scheme to split (\ref{GradientDescent}) into two parts and replace the forcing part $\frac{\partial f}{\partial t}=-\frac{1}{\epsilon^2} \nabla W_{\mathrm{multi}}(f) $ by an associated thresholding.

%====================================================================================

We propose a {\em Modularity MBO scheme} that alternates between the following two primary steps to obtain an approximate solution $f^n:G\rightarrow V^{\hat{n}}$:
\vspace{.2 in}
\begin{description}
\item[Step 1.]\hfill \\
 A gradient-descent process of temporal evolution consists of a diffusion term and an additional balance term:
\begin{align}
	\frac{\partial f}{\partial t}=-(\mathbf{L}f^{(1)},\ldots,\mathbf{L}f^{(\hat{n})}) +\frac{\delta}{\delta f} \left(\gamma \|f-\mathrm{mean}(f)\|^2_{\ell_2}\right)\,. \label{stepone}
\end{align}
We apply this process on $f^n$ with time $\tau_n$, and we repeat it for $\eta$ time steps to obtain $\hat{f}$.

\item[Step 2.] \hfill\\
We threshold $\hat{f}$ from $R^{\hat{n}}$ into $V^{\hat{n}}$:
\begin{align*}
	f^{n+1}_i=\vec{e}_{g_i}\in V^{\hat{n}}\,, ~\mathrm{where}~g_i=\mathrm{argmax}_{\{1\leq l \leq \hat{n}\}}\{\hat{f}_i^{(l)}\}\,.
\end{align*}
This step assigns to $f_i^{n+1}$ the node in $V^{\hat{n}}$ that is the closest to $\hat{f}_i$.
\end{description}
\vspace{.2 in}
%====================================================================================

To solve (\ref{stepone}), we implement a {\em convex-splitting} scheme \cite{DEyre,ConvexS2}. Equation (\ref{stepone}) is the gradient flow of the energy $H_1+H_2$, where $H_1(f):=\frac{1}{2}\sum_{l=1}^{\hat{n}}\langle f^{(l)},\textbf{L}f^{(l)}\rangle$ is convex and $H_2(f):=-\gamma \|f - \mathrm{mean}(f)\|^2_{\ell_2}$ is concave.  In a discrete-time stepping scheme, the convex part is treated implicitly in the numerical scheme, whereas the concave part is treated explicitly. Note that the convex-splitting scheme for gradient-descent equations is an unconditionally stable time-stepping scheme.

%{\bf map: excuse my naivety here: is it "a" convex splitting scheme or "the" scheme; I changed it to 'a' because I assume there is more than one, but please either confirm or tell me I'm wrong} {\bf [Huiyi: I think it's safer to use `a'. even though in some other papers (Bertozzi's and Esedoglu's) they directly call it `convex splitting', no `a' or `the'.}

The discretized time-stepping formula is
\begin{align}
	\frac{\hat{f}-f^n}{\tau_n} &= -\frac{\delta H_1}{\delta f}(\hat{f})-\frac{\delta H_2}{\delta f}(f^n) \nonumber \\
	  &= -(\mathbf{L}\hat{f}^{(1)},\ldots,\mathbf{L}\hat{f}^{(\hat{n})}) + 2\gamma \vec{k} \odot(f^n - \mathrm{mean}(f^n))\,,
\end{align}
where $(\vec{k}\odot f)(n_i):=k_if_i$, $\hat{f}:G\rightarrow \mathbb{R}^{\hat{n}}$, ($k_i$ is the strength of node $n_i$), and $f^n:G\rightarrow V^{\hat{n}}$.
At each 
%iteration
step, we thus need to solve
\begin{align}
	\left((1+\tau_n \textbf{L})\hat{f}^{(1)}, \ldots,(1+\tau_n \textbf{L})\hat{f}^{(\hat{n})}\right) =f^n +2\gamma \tau_n \vec{k}\odot \left[f^n-\mathrm{mean}(f^n)\right]\,. \label{step1}
\end{align}

For the purpose of computational efficiency, we utilize the low-order (leading) eigenvectors (associated with the smallest eigenvalues) of the graph Laplacian $\mathbf{L}$ to approximate the operator $\textbf{L}$. The eigenvectors with higher order are more oscillatory, and resolve finer scale. Leading eigenvectors provide a set of basis to approximately represent graph functions. The more leading eigenvectors are used, the finer scales can be resolved. In the graph-clustering literature, scholars usually use a small portion of leading eigenvectors of $\textbf{L}$ to find useful structural information in a graph \cite{fan, ShiMalik,DiffusionMap,belkin,Newmanspectral}, (note however that some recent work has explored the use of other eigenvectors \cite{Cucuringu}). In contrast, %[this is in contrast to the "small portion]
one typically uses much more modes when solving partial differential equations numerically (e.g., consider a psuedospectral scheme), because one needs to resolve the solution at much finer scales.

Motivated by the known utility and many successes of using leading eigenvectors (and discarding higher-order eigenvectors) in studying graph structure, we project $f$ onto the space of the $N_{\mathrm{eig}}$ leading eigenvectors to approximately solve (\ref{step1}).
Assume that $f^n = \sum_s \phi_s \mathbf{a}^n_s$, $\hat{f} = \sum_s \phi_s \mathbf{\hat{a}}_s$, and $2\gamma \tau_n \vec{k}\odot (f^n-\mathrm{mean}(f^n))=\sum_s \phi_s\mathbf{b}^n_s$, where $\{\lambda_s\}$ are the $N_{\mathrm{eig}}$ smallest eigenvalues of the graph Laplacian $\textbf{L}$.  We denote the corresponding eigenvectors (eigenfunctions) by $\{\phi_s\}$. Note that $\mathbf{a}^n_s$, $\mathbf{\hat{a}}_s$, and $\mathbf{b}^n_s$ all belong to $\mathbb{R}^{\hat{n}}$.  With this representation, we obtain
\begin{align}
	\mathbf{\hat{a}}_s=\frac{\mathbf{a}^n_s+\mathbf{b}^n_s}{1+\tau_n \lambda_s}\,, \quad l\in\{1, 2,\ldots,N_{\mathrm{eig}}\}
\end{align}
from (\ref{step1}) and are able to solve (\ref{step1}) more efficiently.
%\textbf{Input parameters required}: $\gamma, \tau_n,s$, \emph{Neig} (number of eigenvectors used to approximate the Laplacian operator) and a maximum number of iteration (iteration will automatically stop if there is no update on $u_{n+1}$).\\

%%%%%%=== pseudo code====%%%%%%%%%%%%%%%%%%%%%%%%%
We summarize our Modularity MBO scheme in Algorithm \ref{algorithm}. Note that the time complexity of each MBO iteration step is $O(N)$.

\begin{algorithm*}[h]
\begin{algorithmic}
% \begin{enumerate}
\State Set values for $\gamma$, $\hat{n}$, $\eta$, and $\tau_n = dt$.
\State Input $\leftarrow$ an initial function $f^0:G\rightarrow V^{\hat{n}}$ and the 
%$N_{eig}$ smallest 
eigenvalue-eigenvector pairs $\{(\lambda_s, \phi_s)\}$ of the graph Laplacian $\textbf{L}$ corresponding to the $N_{\mathrm{eig}}$ smallest eigenvalues.

\State Initialize:
  \State $\mathbf{a}^0_s=\langle f^0,\phi_s \rangle$;
  \State $\mathbf{b}^0_s =\langle 2\gamma dt \vec{k}\odot(f^0-\mathrm{mean}(f^0)),\phi_s \rangle$.
\While {$f^{n}\not= f^{n-1}$ and $n\leq500$:}
\State Diffusion:
 \For{$i=1 \to \eta$} 
%  \begin{enumerate}
   %\item $a^{n}_s=\langle f^n,\phi_s\rangle$;
      \State $\mathbf{a}^n_s \leftarrow \frac{\mathbf{a}^n_s+\mathbf{b}^n_s}{1+dt \lambda_s}\,,~\mathrm{for }~~s\in\{1, 2,\ldots,N_{\mathrm{eig}}\}$;
   \State $f^n \leftarrow \sum_s \phi_s\mathbf{a}^n_s$;
   \State $\mathbf{b}^{n}_s =\langle 2\gamma dt \vec{k}.*(f^n-\mathrm{mean}(f^n)),\phi_s \rangle$;
   \State i=i+1;
%   \end{enumerate}
\EndFor
  \State Thresholding:
  \begin{align*}
  f^{n+1}_i=\vec{e}_{g_i}\in V^{\hat{n}},~\mathrm{where}~g_i=\mathrm{argmax}_{\{1\leq l \leq \hat{n}\}}\{\hat{f}_i^{(l)}\}.
  \end{align*}
  \State $n=n+1;$
\EndWhile
  \State Output $\leftarrow$ the partition function $f^n$.
%  \end{enumerate}
\end{algorithmic}
\caption{The Modularity MBO scheme.}
\label{algorithm}
  \end{algorithm*}

Unless specified otherwise, the numerical experiments in this paper using a random initial function $f^0$. (It takes its value in $V^{\hat{n}}$ with uniform probability by using the command {\tt rand} in {\sc Matlab}.)
%\footnote{The command is $\vec{v}$={\tt rand}($\hat{n}$, $N$); $f^0(n_i)=\vec{e}_l$, s.t. $\vec{v}(l,i)$={\em max}($\vec{v}(:,i)$).) {\bf [Huiyi: I'm not very sure about the formatting here regarding the matlab code.]}} \blue{I don't think we need to say this.}
%{\bf map: question on above algorithm: "randomly" initialized; can we state somewhere (in the main text, so as not to distract from the algorithm) more specifically what this is; e.g. which command do we use in Matlab?  I appreciate this is standard, but there is more than one choice, so we do need to indicate more so that what we do can be precisely duplicated}

%===================================================================================================

\subsection{Two Implementations of the Modularity MBO Scheme}

Given an input value of the parameter $\hat{n}$, the Modularity MBO scheme partitions a graph into at most $\hat{n}$ communities. In many applications, however, the number of communities is usually not known in advance \cite{MasonNotice,Fortunato09}, so it can be difficult to decide what values of $\hat{n}$ to use. Accordingly, we propose two implementations of the Modularity MBO scheme. The \emph{Recursive Modularity MBO (RMM)} scheme is particularly suitable for networks that one expects a large number of communities, whereas the \emph{Multiple Input-$\hat{n}$ Modularity MBO (Multi-$\hat{n}$ MM)} scheme is particularly suitable for networks that one expects to have a small number of communities.

\textbf{\itshape Implementation 1.} The RMM scheme performs the Modularity MBO scheme recursively, which is particular suitable for networks that one expects to have a large number of communities. In practice, we set the value of $\hat{n}$ to be large in the first round of applying the scheme, and we then let it be small for the rest of the recursion steps. In the experiments that we report in the present paper, we use $\hat{n}=50$ for the first round and $\hat{n}=\mathrm{min}(10, |S|)$ thereafter, where $|S|$ is the size of the subnetwork that one is partitioning in a given step. (We also tried $\hat{n}=10$, $20$ or $30$ for the first round and $\hat{n}=\mathrm{min}(10, |S|)$ thereafter. The results are similar.)

%The choice of $\hat{n}$ does not seem have a significant impact on results.

%{\bf map: last sentence: I assume we are going to be much more specific about that statement later in the paper?  We \emph{need} to do that to be able to include that statement; MAP: coming back to this statement after having gone through the whole paper, I don't think we were ever specific enough about the last sentence; there needs to be additional commentary in the paper somewhere in order to justify this statement}{\bf [Huiyi: can I just delete the last sentence? I tried a few other values of $\hat{n}$, but I guess it's not enough to support the statement.]}

Importantly, the minimization problem (\ref{minimization}) needs a slight adjustment for the recursion steps. Assume for a particular recursion step that we perform the Modularity MBO partitioning with parameter $\hat{n}$ on a network $S \subset G$ containing a subset of the nodes of the original graph.  Our goal is to increase the modularity for the global network instead of the subnetwork $S$. Hence, the target energy to minimize is
\begin{align*}
	H^{(S)}(f) := |f|^{(S)}_{TV} -\gamma \frac{m^{(S)}}{m}\left\|f-\mathrm{mean}^{(S)}(f)\right\|^2_{\ell_2}\,,
\end{align*}
where $f:S\rightarrow V^{\hat{n}} \subset \mathbb{R}^{\hat{n}}$, the TV norm is $|f|^{(S)}_{TV}=\frac{1}{2}\sum_{i,j \in S}w_{ij}|f_i-f_j|_{\ell_1}$, the total edge weight of $S$ is $2m^{(S)} = \sum_{i\in S}k_i$, and $\mathrm{mean}^{(S)}(f) =\frac{1}{2m^{(S)}}\sum_{i \in S} k_if_i$. The rest of the minimization procedures are the same as described previously.
%{\bf map: I am confused by the description above; the adjective "$\hat{n}$-class" for the scheme seems not to have been defined before, at least as concerns the specific terminology used; I assume this is the main scheme introduced above before any implementation discussion; hence, this term needs to be stated explicitly when we first define this scheme}

Note that this recursive scheme is adaptive in resolving the network structure scale. The eigenvectors of the subgroups are recalculated at each recursive step, so the scales being resolved get finer as the recursion step goes. Therefore $N_{\mathrm{eig}}$ need not to be very large.

\textbf{\itshape Implementation 2.} For the Multi-$\hat{n}$ MM scheme, one sets a search range $T$ for $\hat{n}$, runs the Modularity MBO scheme for each $\hat{n}\in T$, and then chooses the resulting partition with the highest modularity score. It works well if one knows the approximate maximum number of communities and that number is reasonably small.  One can then set the search range $T$ to be all integers between 2 and the maximum number. Even though the Multi-$\hat{n}$ MM scheme allows partitions with fewer than $\hat{n}$ clusters, it is still necessary to include small values of $\hat{n}$ in the search range to better avoid local minimums. (See the discussion of the MNIST ``4-9'' digits network in Section \ref{mnist49}.) For different values of $\hat{n}$, one can reuse the previously computed eigenvectors because $\hat{n}$ does not affect the graph Laplacian. Inputting multiple choices for the random initial function $f^0$ (as described at the end of Section \ref{section algorithm}) also helps to reduce the chance of getting stuck in a minimum and thereby to achieve a good optimal solution for the Modularity MBO scheme. Because this initial function is used after the computation of eigenvectors, it only takes a small amount of time to rerun the MBO steps.

%{\bf map: choice of random function: this needs to be discussed \emph{somewhere} in the main text or else we have not given enough information for somebody to duplicate what we've done; please add a short discussion of this (with citations, as relevant) to the relevant section of the paper; right now, this is mysterious to me, and I would absolutely insist on this as a referee}
%{\bf [Huiyi: Mason, in the second last sentence, when I say `helps to avoid', I don't mean totally avoid, but just reduce the chance of being stuck in a local minima. So if my English is wrong and didn't convey this message here, please change it.]}

In Section \ref{numres}, we test these two schemes on several real and synthetic networks.

 %%%%%%%%%%%%=========Implementation============%%%%%%%%%%%%%%%%%%%%%%%%%

\section{Numerical Results}\label{numres}

In this section, we present the numerical results of experiments that we conducted using both synthetic and real network data sets. Unless otherwise specified, our Modularity MBO schemes are all implemented in {\sc Matlab}, (which are not optimized for speed). In the following tests, we set the parameters of the Modularity MBO scheme to be $\eta=5$ and $\tau_n=1$.
%{\bf map: for my own edification, which things were done in other software (and why); maybe I will see the answer to this below?}{\bf [Huiyi: for the R-C procedure to solve for eigenvalues, it is written in C++ code.]}

%%%%%========LFR data=======%%%%%%%%%%%%%%

\subsection{LFR Benchmark}

%%%%%REWRITE THIS INTRODUCTION!!%%%%%%% [map: has this already been done?]
In Ref.~\cite{LFR}, Lancichinetti, Fortunato, and Radicchi (LFR) introduced an eponymous class of synthetic benchmark graphs to provide tougher tests of community-detection algorithms than previous synthetic benchmarks.
%reflect many properties of real networks, was proposed  testing community detection algorithms.   [they only have minimal connection to real networks, so I needed to tone that down a lot; LFR probably claimed that they have such connection, but they are dead wrong
Many real networks have heterogeneous distributions of node degree and community size,
%, whose tails often decay as power laws.  NO! this is almost always false --- see Stumpf and Porter, Science, 2012
so the LFR benchmark graphs incorporate such heterogeneity. They consist of unweighted networks with a predefined set of non-overlapping communities. As described in Ref.~\cite{LFR}, each node is assigned a degree from a power-law distribution with power $\xi$; additionally, the maximum degree is given by $k_{\mathrm{max}}$ and mean degree is $\langle k\rangle$. Community sizes in LFR graphs follow a power-law distribution with power $\beta$, subject to the constraint that the sum of the community sizes must equal the number of nodes $N$ in the network. Each node shares a fraction $1-\mu$ of its edges with nodes in its own community and a fraction $\mu$ of its edges with nodes in other communities.  (The quantity $\mu$ is called the \emph{mixing parameter}.) The minimum and maximum community sizes, $q_{\mathrm{min}}$ and $q_{\mathrm{max}}$, are also specified. We label the LFR benchmark data sets by $(N,\langle k\rangle, k_{\mathrm{max}}, \xi,\beta,\mu,q_{\mathrm{min}},q_{\mathrm{max}})$. The code used to generate the LFR data is publicly available provided by the authors in \cite{LFR}.
%As  a realistic synthetic network, t {\bf map: the LFR benchmark is *NOT* "realistic"!!! no, no, no, no, no! that is dead wrong; rather, it is a tougher benchmark (much tougher, in some cases) than the previous synthetic-network benchmarks

The LFR benchmark graphs has become a popular choice for testing community detection-algorithms, and Ref.~\cite{LFRcomp} uses them to test the performance of several community-detection algorithms.  The authors concluded, for example, that the locally greedy Louvain algorithm \cite{Blondel} is one of the best performing heuristics for maximizing modularity based on the evaluation of the \emph{normalized mutual information} (NMI) (discussed below in this section). Note that the time complexity of this Louvain algorithm is $O(M)$ \cite{Fortunato09}, where $M$ is the number of nonzero edges in the network.
 %{\bf map: 'one of the best' IN WHAT SENSE? be slightly more specific} 
In our tests, we use the GenLouvain code (in {\sc Matlab}) Ref.~\cite{Netwiki}, which is an implementation of a Louvain-like algorithm. The GenLouvain code a modification of the Louvain locally greedy algorithm \cite{Blondel}, but it was not designed to be optimal for speed.
 %and a well established community detection algorithms in the literature. 
 We implement our RMM scheme on the LFR benchmark, and we compare our results with those of running the GenLouvain code. We use the recursive version of the Modularity MBO scheme because the LFR networks used here contain about $0.04N$ communities.

%{\bf map: "LFR networks contain a large number of communities": I have multiple questions here; (1) how large is 'large' ?????? (2) doesn't it not have to be the case that it does?  the ones you check do, but the more general statement that LFR must have that is not true right?  please adjust the statement above to make it precisely correct}

%{\bf map: "GenLouvain" is *not* well-established, so it cannot be described as such; \emph{Louvain} is well-established; I don't think we need to use Louvain itself, BUT if we want to make certain comparative statements about Louvain itself (such as for speed), then we will have to do that; I am ok with not doing those calculations but then also not making those comparisons; }

%but there are statements in the paper that I have needed to change in important ways because of the fact that we haven't done the tests that would be required to justify making them [in particular, speed of genlouvain is unfair to use to make comments regarding speed of louvain]}

We implement the modularity-optimization algorithms on severals sets of LFR benchmark data. 
%We use $\gamma = 1$ in each case but vary the other paremeters.  
%{\bf [map: Huiyi, I couldn't parse your parenthetical comment and assumed that this was what you meant by it; please check this!]} 
We then compare the resulting partitions with the known community assignments of the benchmarks (i.e., the ground truth) by examining the \emph{normalized mutual information} (NMI) \cite{NMI}. %{\bf [Huiyi: LFR1k is averaged over 100 runs, LFR50k is 4 runs. I will mention this later in the subsubsections.]}

Normalized mutual information (NMI) is a similarity measure for comparing two partitions based on the information entropy, and it is often used for testing community-detection algorithms \cite{LFR, LFRcomp}. The NMI equals 1 when two partitions are identical, and it has an expected value of $0$ when they are independent. For an $N$-node network with two partitions, $C=\{C_1,C_2,\ldots,C_K\}$ and $\hat{C}=\{\hat{C}_1,\hat{C}_2,\ldots,\hat{C}_{\hat{K}}\}$, that consist of non-overlapping communities, the NMI is
\begin{equation} 	
	\mathrm{NMI}(C,\hat{C})=\frac{2\sum_{k=1}^K \sum_{\hat{k}=1}^{\hat{K}}P(k,\hat{k})\mathrm{log}\left[\frac{P(k,\hat{k})}{P(k)P(\hat{k})}\right]}{-\sum_{k=1}^K P(k)\mathrm{log}\left[P(k)\right] - \sum_{\hat{k}=1}^{\hat{K}}P(\hat{k}) \mathrm{log}\left[P(\hat{k})\right]}\,,
\end{equation}	
where $P(k,\hat{k})=\frac{|C_k\cap \hat{C}_{\hat{k}}|}{N}$, $P(k)=\frac{|C_k|}{N}$, and $P(\hat{k})=\frac{|\hat{C}_{\hat{k}}|}{N}$.

 \begin{figure}[h!]
\centering
\subfigure[NMI and Modularity ($Q$).]{\includegraphics[width=0.7\textwidth]{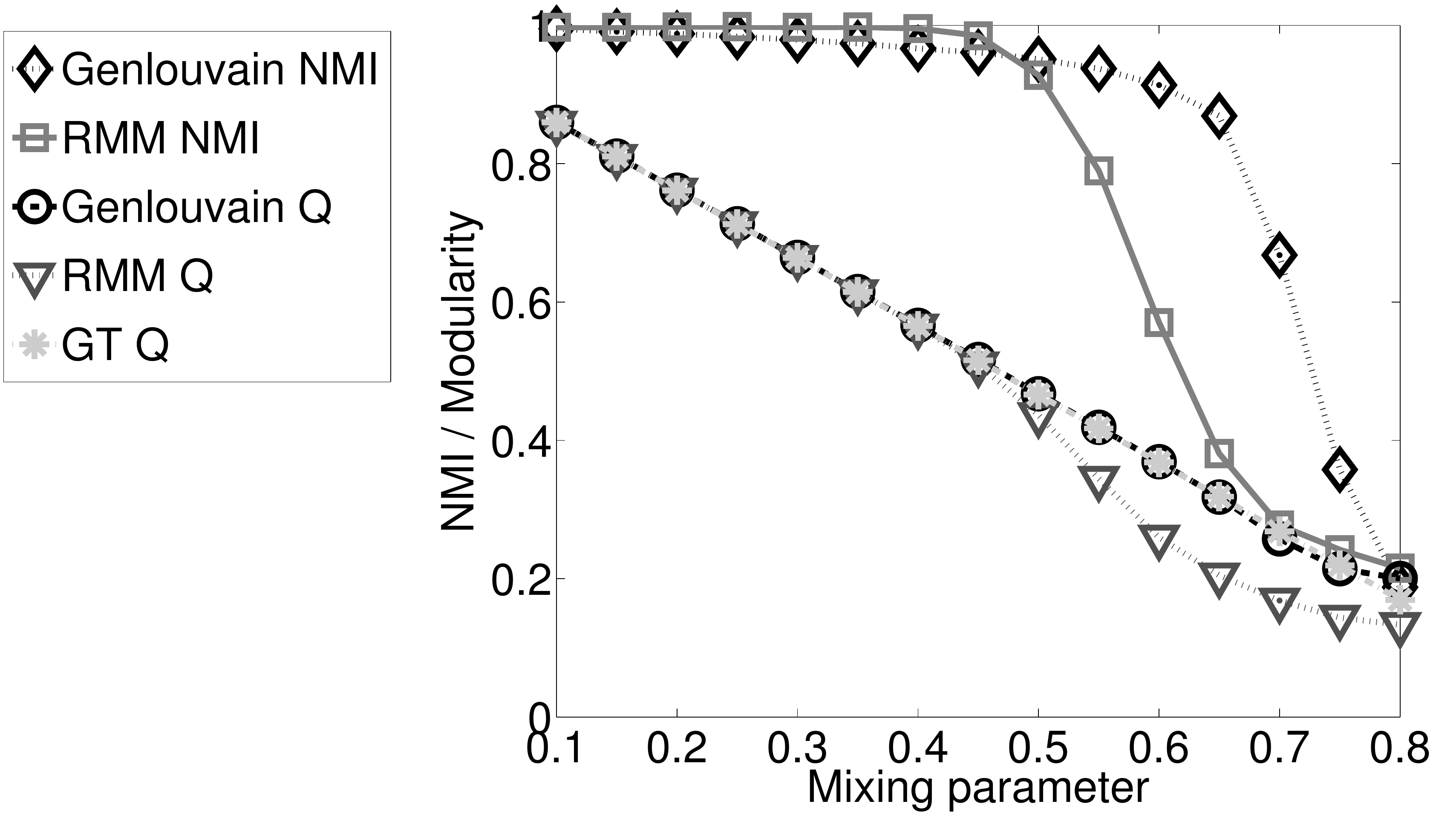}}
\subfigure[Number of Communities ($N_c$).]{\includegraphics[width=0.7\textwidth]{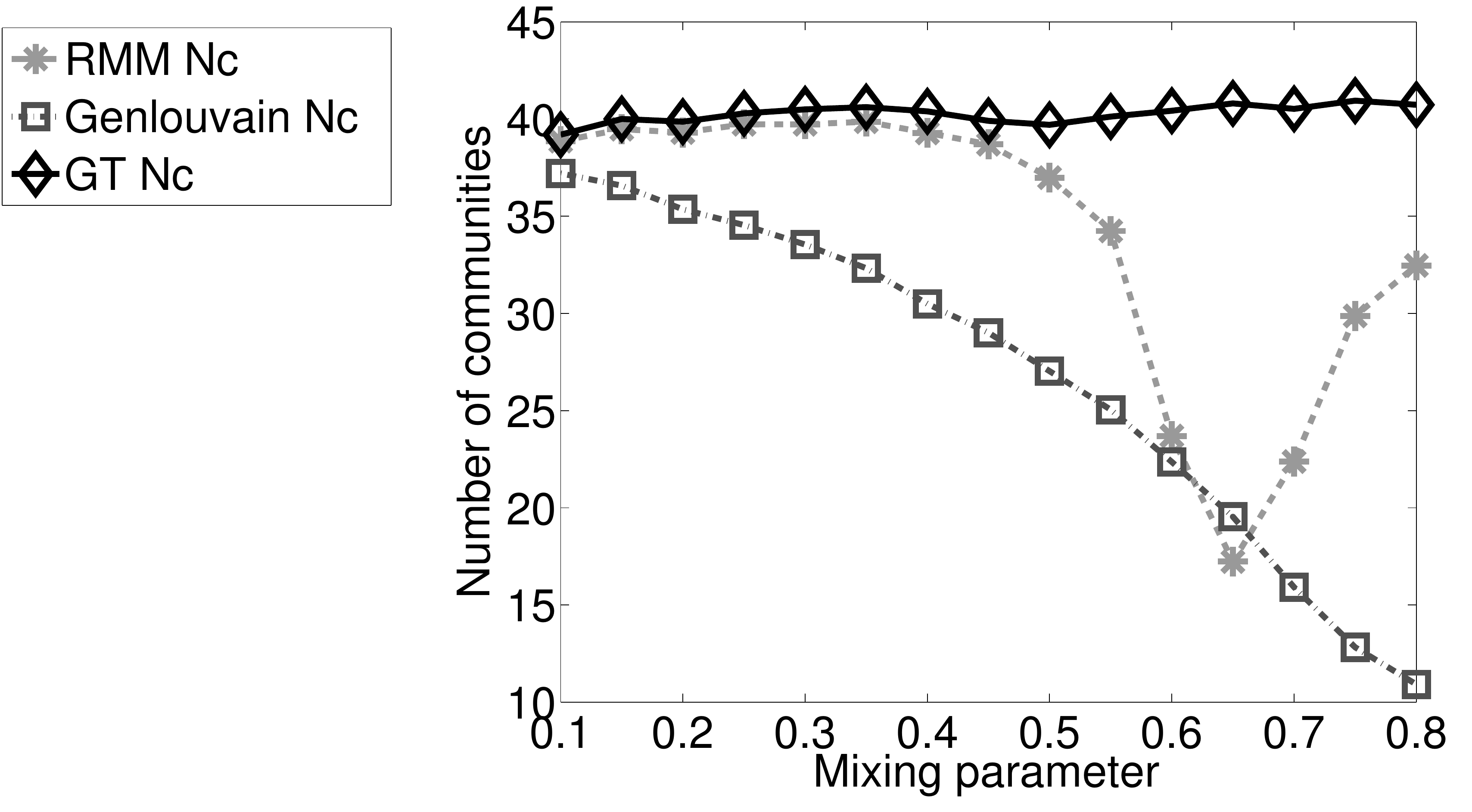}}\\
\caption{Tests on LFR1k networks with RMM and GenLouvain. The ground-truth communities are denoted by GT.
%{\bf map: I assume that "GT" is explained somewhere in the main text?  I don't know what this means}
}
\label{LFRfigure}
\end{figure}

We examine two types of LFR networks. One is the 1000-node ensembles used in Ref.~\cite{LFRcomp}: \[\mathrm{LFR1k}: (1000, 20,50,2,1,\mu,10,50)\,,\] where $\mu\in\{0.1, 0.15,\ldots,0.8\}$. The other is a 50,000-node network, which we call ``LFR50k" and construct as a composition of 50 LFR1k networks. (See the detailed description below.)

%%%%%%%

\subsubsection{LFR1k Networks}\label{1k} 

We use the RMM scheme (with $N_{\mathrm{eig}}=80$) and the GenLouvain code on ensembles of LFR1k$(1000, 20,50,2,1,\mu,10,50)$ graphs with mixing parameters $\mu \in\{0.1, 0.15,\ldots,0.8\}$.  We consider 100 LFR1k networks for each value of $\mu$. The resolution parameter $\gamma$ equals one here.
%{\bf map: huiyi: is the last sentence above correct?}
 
 In Fig.~\ref{LFRfigure}, we plot the mean maximized modularity score ($Q$), the number of communities ($N_c$), and the NMI of the partitions compared with the ground truth (GT) communities as a function of the mixing parameter $\mu$. As one can see from panel (a), the RMM scheme performs very well for $\mu <0.5$. Both its NMI score and modularity score are competitive with the results of GenLouvain. However, for $\mu \geq 0.5$, its performance drops with respect to both NMI and the modularity scores of its network partitions. From panel (b), we see that RMM tends to give partitions with more communities than GenLouvain, and this provides a better match to the ground truth. However, it is only trustworthy for  $\mu <0.5$, when its NMI score is very close to $1$.
 
The mean computational time for one ensemble of LFR1k, which includes 15 networks corresponding to 15 values of $\mu$, is 22.7 seconds for the GenLouvain code and 17.9 seconds for the RMM scheme.  As we will see later when we consider large networks, the Modularity MBO scheme scales very well in terms of its computational time.

%%%%%%%%

 \subsubsection{LFR50k Networks}  

To examine the performance of our scheme on larger networks, we construct synthetic networks (LFR50k) with 50,000 nodes. {To construct an LFR50k network, we start with 50 different LFR1k networks $N_1, N_2,\ldots, N_{50}$ with mixing parameter $\mu$, and we connect each node in $N_s$ ($s\in \{1,2,\ldots,50\}$) to $20\mu $ nodes in $N_{s+1}$ uniformly at random (where we note that $N_{51}=N_1$). We thereby obtain an LFR50k network of size $50,000$. Each community in the original $N_s, s=1,2,\ldots,50$ is a new community in the LFR50k network.} % {\bf [map: Huh?  I don't understand the previous sentence.  Please rewrite it.} 
We build four such LFR50k networks for each value of $\mu=0.1, 0.15,\ldots,0.8$, and we find that all such networks contain about 2000 communities. The mixing parameter of the LFR50k network constructed from LFR1k($\mu$) is approximately $\frac{2\mu}{1+\mu}$. 
 
 %{\bf map: Huiyi: by "randomly", do you mean "uniformly at random" or by some other random process?  simply writing 'randomly' is underspecified!  [I changed this text already because I assume that's what you meant; nevertheless, please confirm}
 
% {\bf [map: 'usually'? what do you mean?  how many different examples of it do we use in this paper?  I thought you only did a computation on one such network; I think you should make this sentence more precise; how many different examples of such a network did you use in the computation, for example? [is my interpretation of what you wrote for the size of those networks correct???]}
 
By construction, the LFR50k network has a similar structure as LFR1k. Importantly, simply increasing $N$ in LFR$(N,\langle k\rangle, k_{\mathrm{max}}, \xi,\beta,\mu,q_{\mathrm{min}},q_{\mathrm{max}})$ to 50,000 is insufficient to preserve similarity of the network structure. A large $N$ results in more communities, so if the mixing parameter $\mu$ is held constant, then the edges of each node that are connected to nodes outside of its community will be distributed more sparsely. In another words, the mixing parameter does not entirely reflect the balance between a node's connection within its own community versus to its connections to other communities, as there is also a dependence on the total number of communities.

The distribution of node strengths in LFR50k is scaled approximately by a factor of $(1+2\mu)$ compared to LFR1k, while the total number of edges in LFR50k is scaled approximately by a factor of $50(1+2\mu)$. Therefore, the probability null model term $\frac{k_ik_j}{2m}$ in modularity (\ref{modularity}) is also scaled by a factor of $\frac{(1+2\mu)}{50}$. {Hence, in order to probe LFR50k with a resolution scale similar to that in LFR1k, it is reasonable to use the resolution $\gamma=50$ to try to minimize issues with modularity's resolution limit \cite{resolution}.} We then implement the RMM scheme ($N_{\mathrm{eig}}=100$) and the GenLouvain code. Note that we also implemented the RMM scheme with $N_{\mathrm{eig}}=500$, but there is no obvious improvement in the result even though there are about $2000$ communities. This is because the eigenvectors of the subgroups are recalculated at each recursive step, so the scales being resolved get finer as the recursion step goes.

%{\bf map: "the node strength in LFR50k is almost the same as in LFR1k": HUH?  I don't understand this statement; the node 'strength' is a statement about a single node, yet you seem to be trying to say something about the full network; I am unable to parse this statement, so I will need you to rephrase it so that I understand the claim that is being made}

%{\bf map: %several additional comments: (1) it looks like this paper uses $\gamma$ in two different ways, so there needs to be a notation change; there is the $\gamma$ in the LFR benchmarks --- which until the paragraph above was the most recent $\gamma$ being used --- and there is the $\gamma$ from the resolution limit; and didn't you specifically say earlier in the paper when you introduced modularity that you were only considering a resolution parameter value of $\gamma = 1$????  this needs to be clearned up, I think --- partly with exposition and partly by using different symbols for these two different objects; (2) the above does not "avoid" the resolution limit; I can see how it will make it less of a probably, but it is not justified to state that we "avoid" the issue; I think my adjusted phrasing above is still not good enough, so 
%I think we need to think about this phrasing}{\bf [Huiyi: Mason, please check the part in bold.]}

 \begin{figure}[h!]
\centering
\subfigure[NMI and Modularity ($Q$).]{\includegraphics[width=0.7\textwidth]{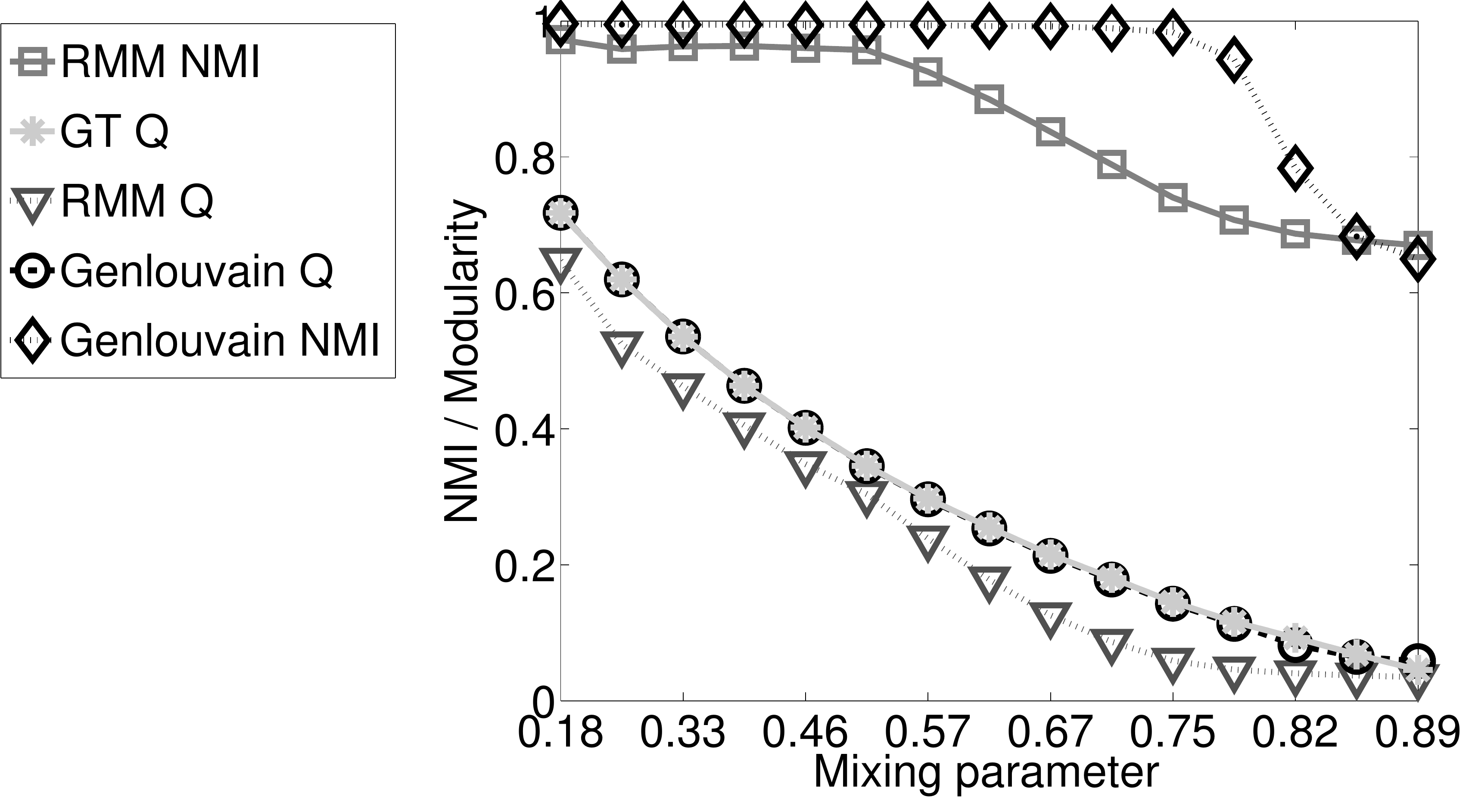}}\\
\subfigure[Number of Communities ($N_c$).]{\includegraphics[width=0.7\textwidth]{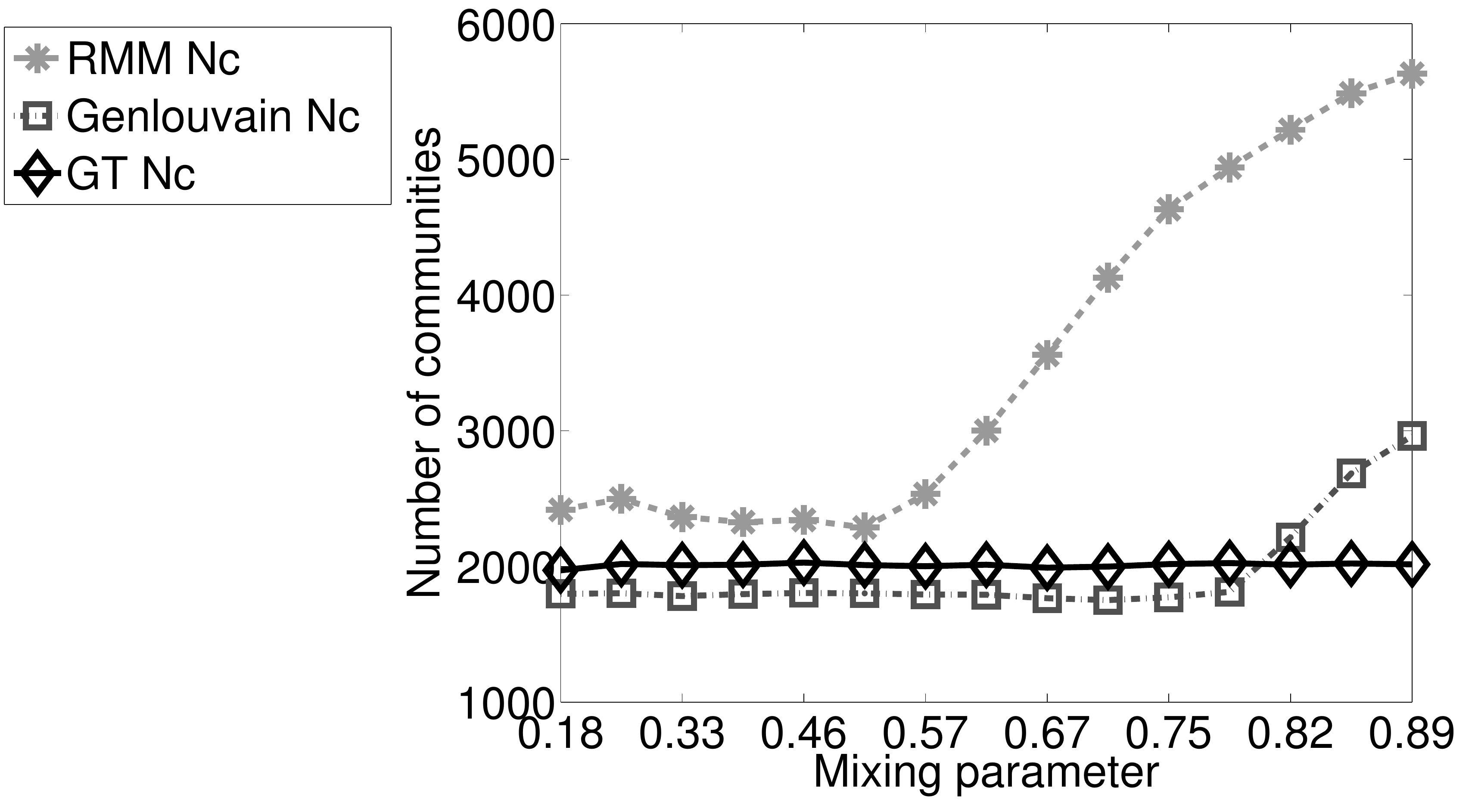}}
\caption{Tests on LFR50k data with RMM and GenLouvain.
%{\bf map: this figure (and liking others): the capitalization of the x-axis labels is not using the same conventioning as the y-axis labels; this needs to be consistent, so what this should be is "Mixing Parameter" to go with "Number of Clusters"; please also check and fix this for all other figures}
}
\label{LFR50k}
\end{figure}

We average the network diagnostics over the four LFR50k networks for each value of mixing parameter.
In Fig.~\ref{LFR50k}, we plot the network diagnostics 
%similarly as in Section \ref{1k} 
versus the mixing parameter $\frac{2\mu}{1+\mu}$ for $\mu \in \{0.1,0.15,\ldots,0.8\}$. 
%From Figure \ref{LFR50k} (a) 
In panel (a), we see that the performance of RMM is good only when the mixing parameter is less than 0.5, though it is not as good as GenLouvain. It seems that the recursive Modularity MBO scheme has some difficulties in dealing with networks with very large number of clusters. 

However the computational time of RMM is lower than that of the GenLouvain code \cite{Netwiki} (though we note that it is an implementation that was not optimized for speed). The mean computational time for an ensemble of LFR50k networks, which includes 15 networks corresponding to 15 values of $\mu$, is 690 seconds for GenLouvain and 220 seconds for the RMM scheme. In Table \ref{table LFR}, we summarize the mean computational time (in seconds) on each ensemble of LFR data.

%{\bf map: important point to discuss: genlouvain is not optimized for speed; it can be made faster; it was released to give people something hopefully useful, not with an attempt to make it as fast as possible; so as long as we include such a statement it's ok to write the statements about speed; also: did you do your calculations on GenLouvain 1.2 or on an earlier version?  originally I gave you an earlier version; v1.2 is \emph{much} faster than prior versions, so is that the version of GenLouvain that you used for this calculation or did you use a prior one? we need to state explicitly which version you used; more generally, we need to be really sure that we don't accidentally set up a straw man with the claims that we make; in particular, I think the advantages of the new method in this paper that we need to stress are not the ones that rely on the speed of a method whose implementation can in principle be made faster without adjusting the algorithm at all; if we're going to stress the speed, the comparison has to be something that is established as a fast *implementation*, which is not the case for GenLouvain --- take a look at the table in \cite{Fortunato09}; perhaps we can look at the computational complexities listed therein and use that to help with a discussion?}

 \begin{table}[h!]
\centering
\begin{tabular}{ |c | c|c|  }
  \hline                        
   &LFR1k &LFR50k\\
 \hline 
  GenLouvain&22.7 s&690 s\\
   \hline
  RMM&17.9 s&220 s\\
   \hline
  \end{tabular} 
  \caption{}
  \label{table LFR}
 \end{table}

%{\bf map: important note: need to state explicitly in comparison what compares to our implementation rather than the algorithm itself, as our implementations aren't necessarily optimized for speed; need to be very careful about how we phrase such things!}

%%%%%========MNIST=======%%%%%%%%%%%%%%

\subsection{MNIST Handwritten Digit Images}

The MNIST database consists of 70,000 images of size $28\times28$ pixels containing the handwritten digits ``0" through ``9" \cite{MNISTdata}. The digits in the images have been normalized with respect to size and centered in a fixed-size grey image. In this section, we use two networks from this database. We construct one network using all samples of the digits ``4" and digit ``9", which are difficult to distinguish from each other and which constitute 13782 images of the 70000. We construct the second network using all images.
 %samples containing digits from 0 to 9. 
 In each case, our goal is to separate the distinct digits into distinct communities.
 %of each digit out of the whole network.

%{\bf map: mnist data: can we please also cite the first paper that analyzed this in addition to citing the page where the data is available}

%{\bf [Huiyi: Thomas, which ref do you recommend to cite? It would be nice if this ref also describes the difficulty of separating digit 4 and 9.]}

We construct the adjacency matrices (and hence the graphs) $\textbf{W}$ of these two data sets as follows. First, we project each image (a $28^2$-dimensional datum) onto 50 principal components.  For each pair of nodes $n_i$ and $n_j$ in the 50-dimensional space, we then let $w_{ij}=\exp\left(-\frac{d_{ij}^2}{3\sigma^2}\right)$ if either $n_i$ is among the 10 nearest neighbors of $n_j$ or vice versa; otherwise, we let $w_{ij}=0$. The quantity $d_{ij}$ is the $\ell_2$ distance between $n_i$ and $n_j$, the parameter $\sigma$ is the mean of distances between $n_i $ and its 10th nearest neighbor.

In this data set, the maximum number of communities is 2 when considering only the digits ``4" and ``9", and it is 10 when considering all digits.  We can thus choose a 
%reasonably 
small search range for $\hat{n}$ and use the Multi-$\hat{n}$ Modularity MBO scheme.

%%%%%%%%

\subsubsection{MNIST ``4-9'' Digits Network}\label{mnist49} 

This weighted network has 13782 nodes and 194816 weighted edges. We use the labeling of each digit image as the ground truth.  There are two groups of nodes: ones containing the digit ``4" and ones containing the digit ``9".  We use these two digits because they tend to look very similar when they are written by hand.  In Fig.~\ref{Mnistfigure}(a), we show a visualization of this network, where we have projected the data projected onto the second and third leading eigenvectors of the graph Laplacian $\textbf{L}$. The difficulty of separating the ``4" and ``9" digits has been observed in the graph-partitioning literature (see, e.g., Ref.~\cite{pro:HeinSetzer11TightCheeger}). For example, there is a near-optimal partition of this network using traditional spectral clustering \cite{tutorialSC,ShiMalik} (see below) that splits both the ``4''-group and the ``9''-group roughly in half.   

The modularity-optimization algorithms that we discuss for the ``4-9''  network use $\gamma=0.1$. We choose this resolution-parameter value so that the network is partitioned into two groups by the GenLouvain code. The question about what value of $\gamma$ to choose is beyond the scope of this paper, but it has been discussed at some length in the literature on modularity optimization \cite{Fortunato09}. Instead, we focus on evaluating the performance of our algorithm with the given  value of the resolution parameter. We implement the Modularity MBO scheme with $\hat{n} = 2$ and the Multi-$\hat{n}$ MM scheme, and we compare our results with that of the GenLouvain code as well as traditional spectral clustering method \cite{tutorialSC,ShiMalik}. 
%{\bf map: why do we suddenly have $\gamma = 0.1$.  skeptical review says: are we cherry-picking this?  we need some sort of comment; also, see earlier comment about the $\gamma$ values we're using and stuff we wrote in the beginning of the article}

%{\bf map: we use different resolution-parameter values later, right?}

Traditional spectral clustering is an efficient clustering method that has been used widely in computer science and applied mathematics because of its simplicity. It calculates the first $k$ nontrivial eigenvectors $\phi_1,\phi_2,\ldots,\phi_k$ (corresponding to the smallest eigenvalues) of the graph Laplacian $\textbf{L}$. Let $U\in \mathbb{R}^{N\times k}$ be the matrix containing the vectors $\phi_1,\phi_2,\ldots,\phi_k$ as columns. For $i\in \{1,2,\ldots,N\}$, let $y_i\in \mathbb{R}^k$ be the $i$th row vector of $U$. Spectral clustering then applies the $k$-means algorithm to the points $(y_i)_{\{i=1,\ldots,N\}}$ and partitions them into $k$ groups, where $k$ is the number of clusters that was specified beforehand.

%{\bf map: "K means": am I missing something?  you want "$k$" right because you already previously specified which value you want?  I'm confused about the switch to capitals, as if $k$ and $K$ might be different}

On this MNIST ``4-9'' digits network, we specify $k=2$ and implement spectral clustering to obtain a partition into two communities.  As we show in Fig.~\ref{Mnistfigure}(b), we obtain a near-optimal solution that splits both the ``4''-group and the ``9''-group roughly in half.  This differs markedly from the ground-truth partition in panel (a).

 \begin{figure}[h!]
\centering
\subfigure[Ground Truth]{\includegraphics[width=0.45\textwidth]{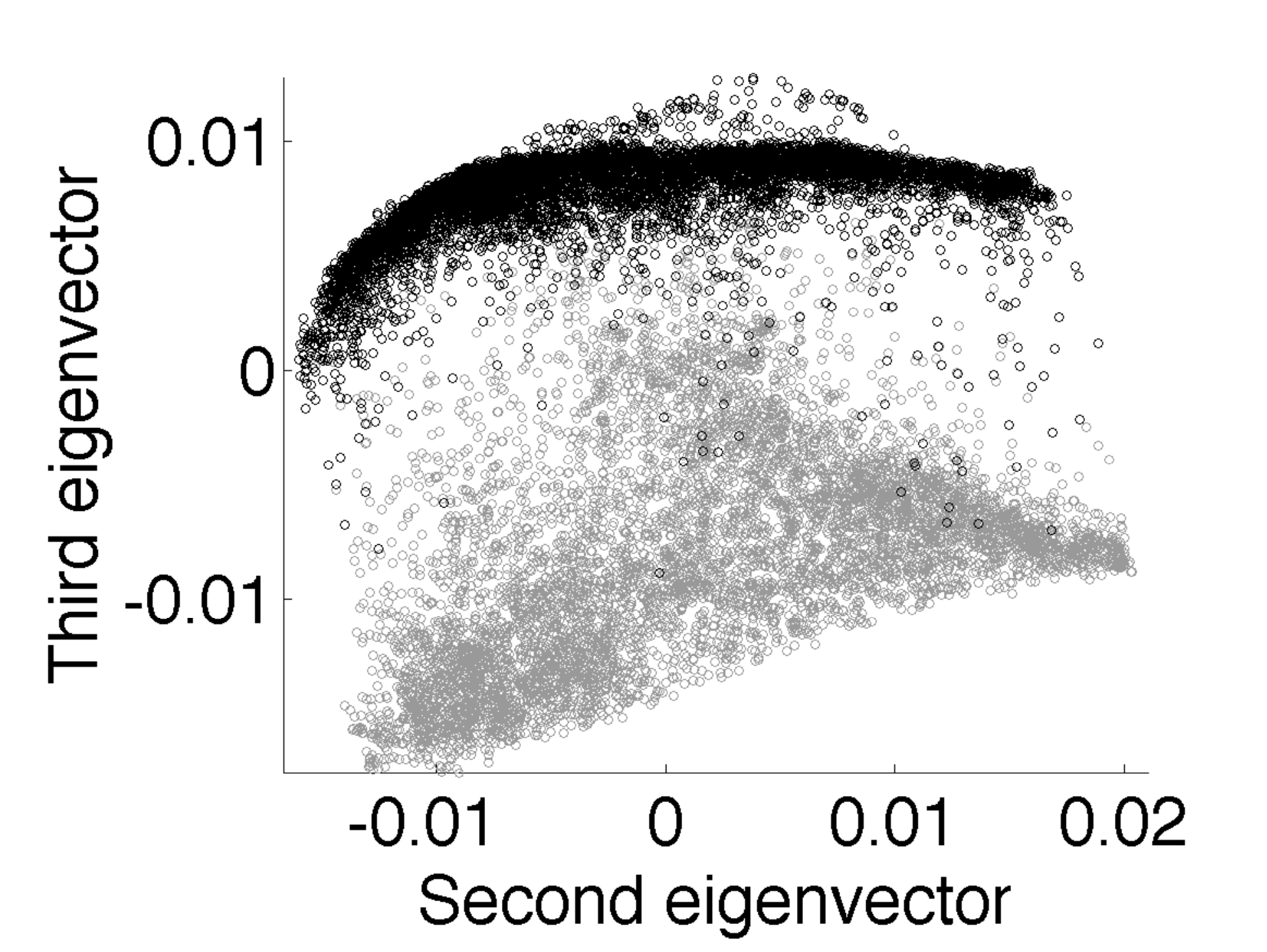}}
\subfigure[Spectral Clustering with $k$-Means]{\includegraphics[width=0.45\textwidth]{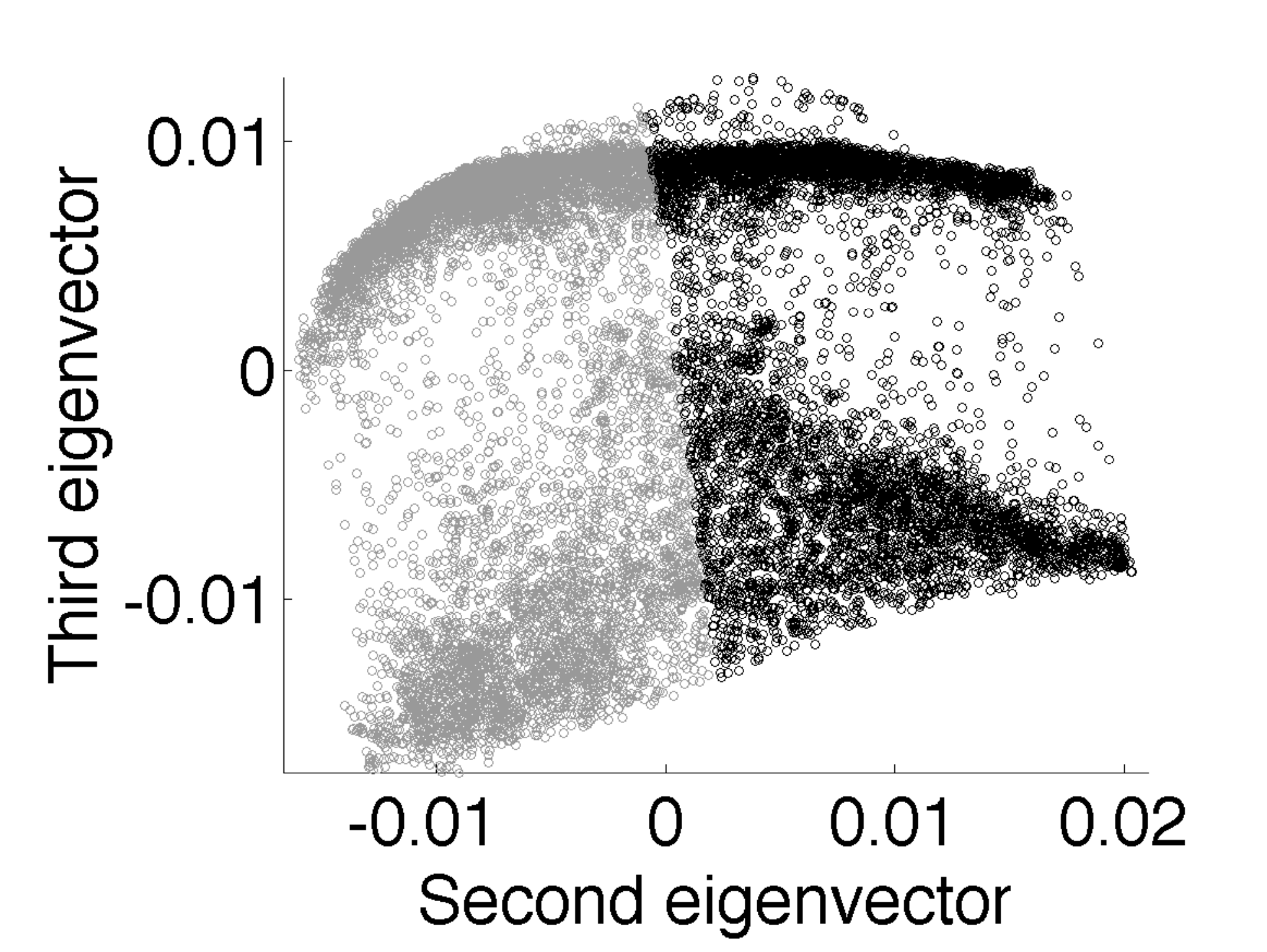}}\\
\subfigure[Modularity MBO Scheme with $\hat{n}$=2]{\includegraphics[width=0.45\textwidth]{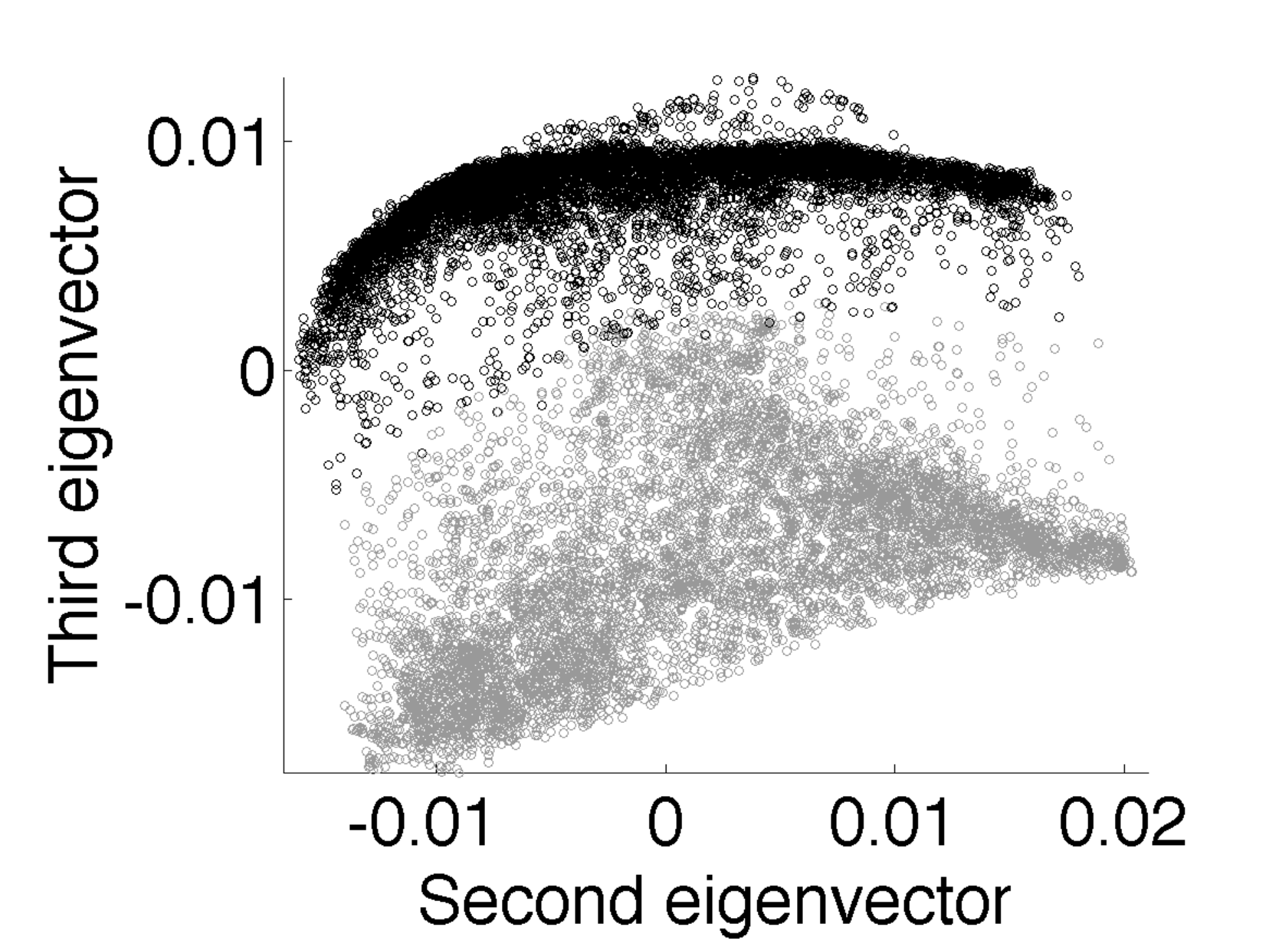}}
\subfigure[Modularity MBO Scheme with $\hat{n}=8$]{\includegraphics[width=0.45\textwidth]{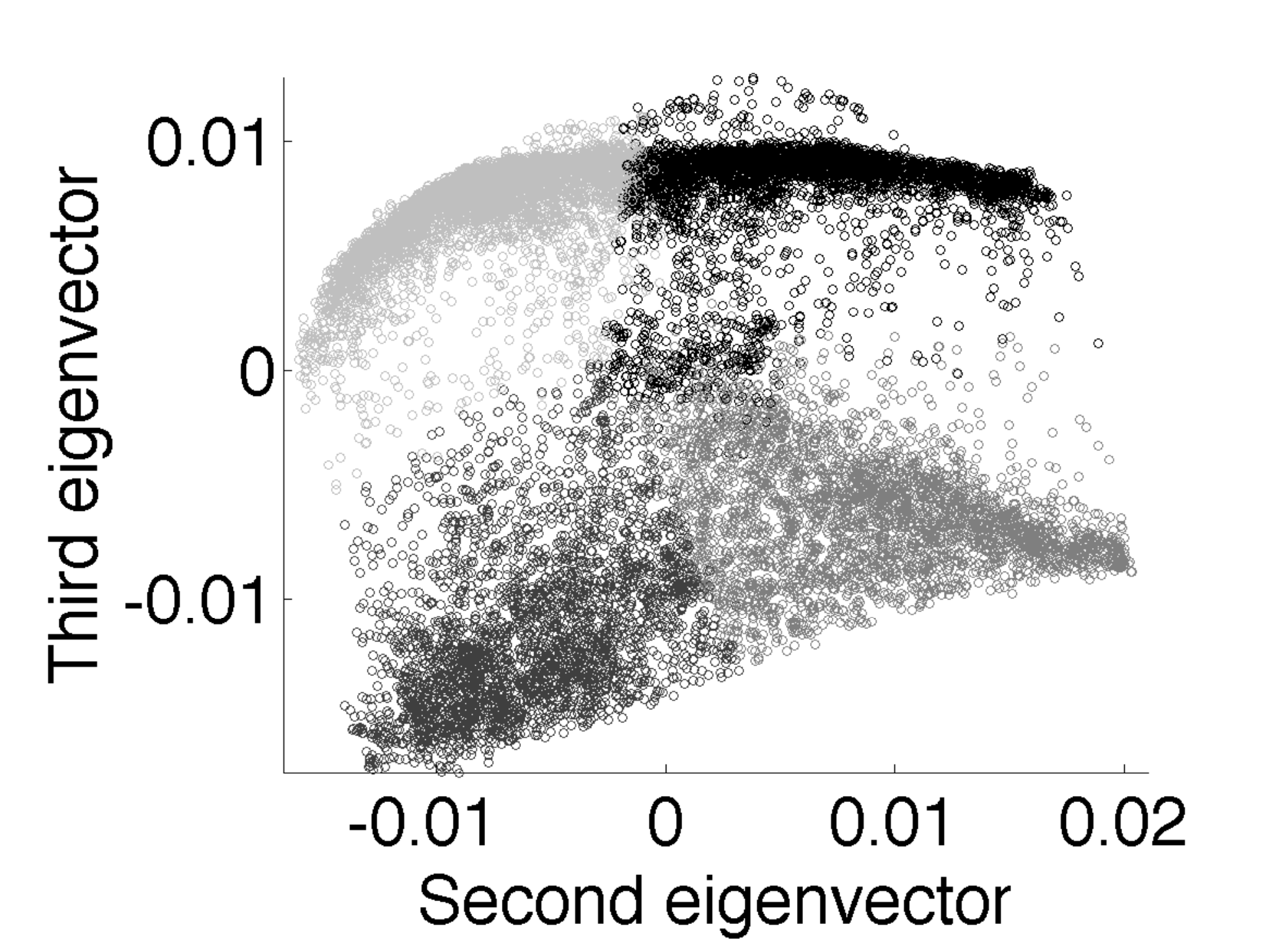}}\\
\subfigure[Modularity Score]{\includegraphics[scale=0.26]{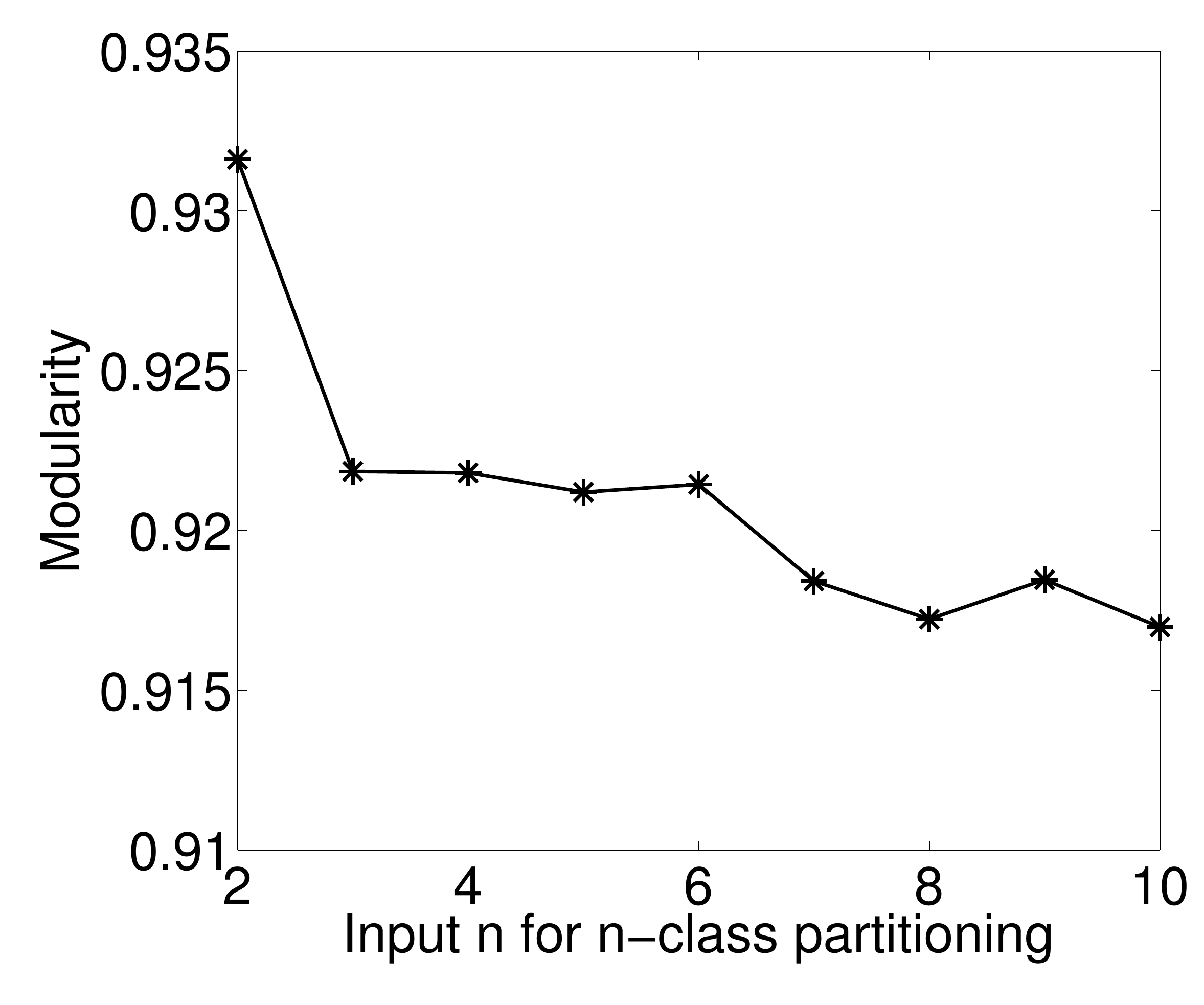}}
\subfigure[Group Assignments Visualization]{\includegraphics[scale=0.26]{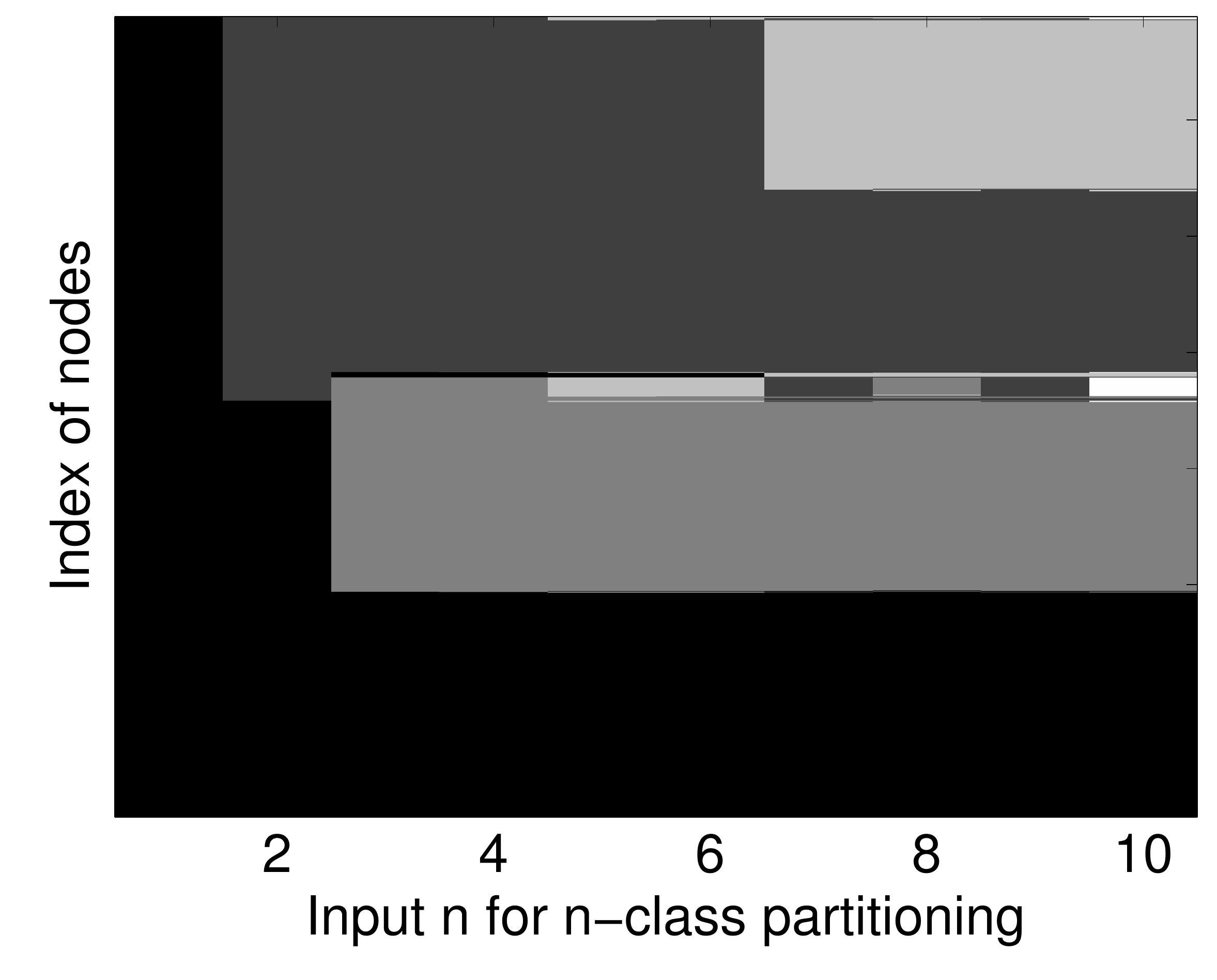}}
\label{Mnistfigure}
\caption{(a)--(d) Visualization of partitions on the MNIST ``4-9" digit image network by projecting it onto the second and third leading eigenvectors of the graph Laplacian. Shading indicates the community assignment. (e)--(f) Implementation results of the Multi-$\hat{n}$ Modularity MBO scheme on the MNIST ``4-9" digit images. In panel (a), shading indicates the community assignment. The horizontal axis represents the input $\hat{n}$ (i.e., the maximum number of communities), and the vertical axis gives the (sorted) index of nodes. In panel (b), we plot the optimized modularity score as a function of the input $\hat{n}$. 
}
%\label{MNIST}
\end{figure}

For the Multi-$\hat{n}$ MM scheme, we use $N_{\mathrm{eig}}=80$ and the search range $\hat{n}\in \{2,3,\ldots,10\}$. We show visualizations of the partition at $\hat{n}=2$ and $\hat{n} = 8$ in Figs.~\ref{Mnistfigure}(c,d). For this method, computing the spectrum of the graph Laplacian takes a significant portion of the run time (9 seconds for this data set).  Importantly, however, this information can be reused for multiple $\hat{n}$, which saves time. In Fig.~\ref{Mnistfigure}(e), we show a plot of this method's optimized modularity scores versus $\hat{n}$. Observe that the optimized modularity score achieves its maximum when we choose $\hat{n}=2$, which yields the best partition that we obtain using this method. In Fig.~\ref{Mnistfigure}(f), we show how the partition evolves as we increase the input $\hat{n}$ from 2 to 10. At $\hat{n}=2$, the network is partitioned into two groups (which agrees very well with the ground truth). For $\hat{n}>2$, however, the algorithm starts to pick out worse local optima, and either ``4''-group or the ``9''-group gets split roughly in half. Starting from $\hat{n}=7$, the number of communities stabilizes at about 4 instead of increasing with $\hat{n}$. This indicates that the Modularity MBO scheme allows one to obtain partitions with $N_c\leq\hat{n}$.
% as mentioned before. 

%{\bf map: Huiyi: you wrote that it "agrees with" the ground truth; I would only use that statement if literally EVERY node is classified in the same way; is that the case?  if so, we can change the phrasing back}

In Table \ref{table mnist49}, we show computational time and some network diagnostics for all of the resulting partitions. The modularity of the ground truth is $Q_{GT} \approx 0.9277$. Our schemes obtain high modularity and NMI scores that are comparable to those obtained using the GenLouvain code (which was not intended by its authors to be optimized for speed). The number of iterations for the Modluarity MBO scheme ranges approximately from 15 to 35 for $\hat{n}\in\{2,3,\ldots,10\}$.

%{\bf map: regular spectral code: Huiyi, this is your own implementation, yes?  and I assume that you have checked that the implementation is such that the speed scaling is optimal compared to the theoretical best of this algorithm?}
%{\bf Huiyi: In matlab there is a command "kmeans" designed for kmeans method. So i think we can assume it is optimal in matlab. }
%{\bf map: I think we need to be very careful about the speeds we're reporting; I am including various caveats regarding what is fair to interpret from that part of the results; I think this is something that we should discuss as a group}

\begin{table}[h!]
\centering
\begin{tabular}{| c | c | c |c|c|c| }
  \hline                        
   & $N_c$ &$Q$&NMI&Purity&Time (seconds) \\
 \hline 
  GenLouvain& 2&0.9305&0.85&0.975&110 s\\
   \hline
   Modularity MBO ($\hat{n}=2$)&2&0.9316&0.85&0.977&11 s\\
   \hline
   Multi-$\hat{n}$ MM ($\hat{n}\in \{2,3,\ldots,10\}$)&2&0.9316&0.85&0.977&25 s\\
   \hline
   Spectral Clustering ($k$-Means)&2&NA&0.003&0.534&1.5 s \\
   \hline
  \end{tabular}
\caption{}
\label{table mnist49}
\end{table}

The {\em purity} score, which we also report in Table \ref{table mnist49}, measures the extent to which a network partition matches ground truth.  Suppose that an $N$-node network has a partition $C=\{C_1,C_2,\ldots,C_K\}$ into non-overlapping communities and that the ground-truth partition is $\hat{C}=\{\hat{C}_1,\hat{C}_2,\ldots,\hat{C}_{\hat{K}}\}$. The purity of the partition $C$ is then defined as
\begin{equation}	
	\mathrm{Prt}(C,\hat{C})=\frac{1}{N}\sum_{k=1}^K \mathrm{max}_{l\in \{1,\ldots,\hat{K}\}} |C_k \cap \hat{C}_l| \in [0,1]\,.
\end{equation}	
Intuitively, purity can by viewed as the fraction of nodes that have been assigned to the correct community. However, the purity score is not robust in estimating the performance of a partition. When the partition $C$ breaks the network into communities that consist of single nodes, then the purity score achieves a value of $1$. hence, one needs to consider other diagnostics when interpreting the purity score. In this particular data set, a high purity score does indicate good performance because the ground truth and the partitions each consist of two communities.

%{\bf map: Huiyi, you wrote "other network properties"; I don't think you mean 'properties' but rather diagnostics that are computed based on the partitions --- not just the networks}

Observe in Table \ref{table mnist49} that all modularity-based algorithms identified the correct community assignments for more than 97\% of the nodes, whereas standard spectral clustering was only correct for just over half of the nodes. The Multi-$\hat{n}$ MM scheme takes only 25 seconds.  If one specifies $\hat{n}=2$, then the Modularity MBO scheme only takes 11 seconds.  
%(These are both much faster than the GenLouvain code \cite{Netwiki}, which we note was never designed to be optimal for speed.) {\bf Huiyi: Mason, it has just been noted right before Table \ref{results}. Maybe it would be enough to just note this once?}
%map: that is only ok if we don't make the explicit comment that they're both much faster than the GenLouvain code

%%%%%%%%

\subsubsection{MNIST 70k Network}

We test our new schemes further by consider the entire MNIST network of 70,000 samples containing digits from ``0" to ``9". This network contains about five times as many nodes as the MNIST ``4-9" network. However, the node strengths in the two networks are very similar because of how we construct the weighted adjacency matrix.  We thus choose $\gamma=0.5$ so that the modularity optimization is performed at a similar resolution scale in both networks. There are 1001664 weighted edges in this network.
%the MNIST 4-9 network.

We implement the Multi-$\hat{n}$ MM scheme with $N_{\mathrm{eig}}=100$ and the search range $\hat{n}\in\{2,3,\ldots,20\}$. Even if $N_c$ is the number of communities in the true optimal solution, the input $\hat{n}=N_c$ might not give a partition with $N_c$ groups. The modularity landscape in real networks is notorious for containing a huge number of nearly degenerate local optima (especially for values of modularity $Q$ near the globally optimum value) \cite{good2010}, so we expect the algorithm to yield a local minimum solution rather than a global minimum. Consequently, it is preferable to extend the search range to $\hat{n} > N_c$, so that the larger $\hat{n}$ gives more flexibility to the algorithm to try to find the partition that optimizes modularity.

%{\bf map: consistency of notation: $\hat{n}=2,3,\ldots,20$ vs $\hat{n}\in\{2,3,\ldots,20\}$.; I think we need to be consistent and only use one of these; note that this is not the only example in this paper where we have such inconsistency; Huiyi, I think you need to go through this type of thing and make sure that we are consistent in the paper}

The best partition that we obtained using the search range $\hat{n}\in\{2,3,\ldots,20\}$ contains 11 communities. All of the digit groups in the ground truth except for the ``1''-group are correctly matched to those communities. In the partition, the ``1''-group splits into two parts, which is unsurprising given the structure of the data.  In particular, the samples of the digit ``1" include numerous examples that are written like a ``7''. This set of samples are thus easily disconnected from the rest of ``1''-group. If one considers these two parts as one community associated with ``1''-group, then the partition achieves a 96\% correctness in its classification of the digits.

As we illustrate in Table \ref{table mnist70k}, the GenLouvain code yields comparably successful partitions as those that we obtained using the Multi-$\hat{n}$ MM scheme.  
%Unsurprisingly, it is much slower as the size of the network is scaled up. (Although we note that the GenLouvain code is not designed to be optimal for speed, neither is ours. \red{[Thomas: saying it once is enough?}]) 
By comparing the running time of the Multi-$\hat{n}$ MM scheme on both MNIST networks, one can see that our algorithm scales well in terms of speed when the network size increases. While the network size increases five times ($5\times$) and the search range gets doubled ($2\times$), the computational time increases by a factor of $11.6\approx 5\times2$.

%{\bf map: um: where does $2 \times 5$ come from?  I can see that in certain contains it would be ok for 11.6 to be approximately 10, but where does the decomposition into $2 \times 5$ come from?  don't we need to write a sentence explaining this?  at minimum, I don't understand it!  so this can't stay as is}

The number of iterations for the Modluarity MBO scheme ranges approximately from 35 to 100 for $\hat{n}\in\{2,3,\ldots,20\}$. Empirically, even though the total number of iterations can be as large as over a hundred, the modularity score quickly gets very close to its final value within the first 20 iteration. 

%The network diagnostics and time costs are listed in the table below. Again our scheme performs as good as the GenLouvain, while running significantly faster (by a factor of about 40--50).

%{\bf map: if we want to do a comparison to Louvain and indicate the time-cost strength of the new methos, then I think we need to comment on the speed of Louvain itself and not an implementation that was never written for the purpose of optimizing speed.  The way the phrasing was before, I feel like we set up a straw man.  I have attempted to change the phrasing to avoid that.  I think the question becomes whether we want the weaker phrasing in that respect or whether we want to do the calculations to be able to make such a comparison fairly.  I don't think what we wrote before is fair to the authors of that code.}

The computational cost of the Multi-$\hat{n}$ MM scheme consists of two parts: the calculation of the eigenvectors and the MBO iteration steps. Because of the size of the MNIST 70k network, the first part costs about 90 seconds in {\sc Matlab}. However, one can incorporate a faster eigenvector solver, such as the Rayleigh-Chebyshev (RC) procedure of \cite{RC}, to improve the computation speed of an eigen-decomposition. This solver is especially fast for producing a small portion (in this case, ${1}/{700}$) of the leading eigenvectors for a sparse symmetric matrix. Upon implementing the RC procedure in C++ code, it only takes 12 seconds to compute the 100 leading eigenvector-eigenvalue pairs. Once the eigenvectors are calculated, they can be reused in the MBO steps for multiple values of $\hat{n}$ and different initial functions $f^0$.  This allows good scalability, which is a particularly nice feature of using this MBO scheme.

\begin{table}[h!]
\centering
\begin{tabular}{ |c | c | c |c|c|c |}
  \hline                        
   & Nc &Q&NMI&Purity&Time (second) \\
 \hline 
  GenLouvain& 11&0.93&0.916&0.97&10900 s\\
   \hline
   Multi-$\hat{n}$ MM ($\hat{n}\in \{2,3,\ldots,20\}$)&11&0.93&0.893&0.96&290 s / 212 s*\\
      \hline
Modularity MBO 3\% GT ($\hat{n}=10$)&10&0.92&0.95&0.96&94.5 s / 16.5 s*\\
\hline
  \end{tabular}
\small{\emph{$^*$Calculated with the RC procedure.}}\\
\caption{}\label{table mnist70k}
\end{table}

Another benefit of the Modularity MBO scheme is that it allows the possibility of incorporating a small portion of the ground truth in the modularity optimization process.  
%To our knowledge, this has not previously been done in studies of modularity optimization.
In the present paper, we implement the Modularity MBO using 3\% of the ground truth by specifying the true community assignments of 2100 nodes, which we chose uniformly at random in the initial function $f^0$. We also let $\hat{n}=10$. With the eigenvectors already computed (which took 12 seconds using the RC process), the MBO steps take a subsequent 4.5 seconds to yield a partition with exactly 10 communities and 96.4\% of the nodes classified into the correct groups. The authors of Ref.~\cite{MBO-multiGL} also implemented a segmentation algorithm on this MNIST 70k data with 3\% of the ground truth, and they obtained a partition with a correctness 96.9\% in 15.4 seconds. In their algorithm, the ground truth was enforced by adding a quadratic fidelity term to the energy functional (semi-supervised). The fidelity term is the $\ell_2$ distance of the unknown function $f$ and the given ground truth. In our scheme, however, it is only used in the initial function $f^0$. Nevertheless, it is also possible to add a fidelity term to the Modularity MBO scheme and thereby perform semi-supervised clustering.

%%%%%========Net-Sci=======%%%%%%%%%%%%%%

\subsection{Network-Science Coauthorships}

Another well-known graph in the community detection literature is the network of coauthorships of network scientists.  This benchmark was compiled by Mark Newman and first used in Ref.~\cite{Newmanspectral}.  

In the present paper, we use the graph's largest connected component, which consists of 379 nodes representing authors and 914 weighted edges that indicate coauthored papers.  We do not have any so-called ground truth for this network, but it is useful to compare partitions obtained from our algorithm with those obtained using more established algorithms. In this section, we use GenLouvain's result as this pseudo-ground truth.  In addition to Modularity-MBO, RMM, and GenLouvain, we also consider the results of modularity-based spectral partitioning methods that allow the option of either bipartitioning or tripartitioning at each recursive stage \cite{Newmanspectral,tripart}..

%\red{[Thomas: Do we need to detail their method?]} 
In Ref.~\cite{Newmanspectral}, Newman proposed a spectral partitioning scheme for modularity optimization by using the leading eigenvectors (associated with the largest eigenvalues) of a so-called {\em modularity matrix} $\textbf{B}=\textbf{W}-\textbf{P}$ to approximate the modularity function $Q$. In the modularity matrix, $\textbf{P}$ is the probability null model and $P_{ij}=\frac{k_i k_j}{2m}$ is the NG null model with $\gamma=1$. Assume that one uses the first $p$ leading eigenvectors $\{\textbf{u}_1, \textbf{u}_2, \ldots, \textbf{u}_p\}$, and let $\beta_j$ denote the eigenvalue of $\textbf{u}_j$ and $\textbf{U}=(\textbf{u}_1|\textbf{u}_2|\ldots | \textbf{u}_p)$. We then define $N$ node vectors $\textbf{r}_i\in \mathbb{R}^p$ whose $j$th component is 
 \begin{equation*}	
	(\textbf{r}_i)_j = \sqrt{\beta_j-\alpha}U_{ij}\,,
\end{equation*}	
where $\alpha \leq \beta_p$ and $j\in\{1,2,\ldots,p\}$. The modularity $Q$ is therefore approximated as 
 \begin{align}
	 Q\simeq \hat{Q}= N\alpha + \sum_{l=1}^{\hat{n}}\|\textbf{R}_l\|_{\ell_2}^2\,, \label{approxQ}
 \end{align}
 where $\textbf{R}_l=\sum_{g_i=l}\textbf{r}_i$ is sum of all node vectors in the $l$th community (where $l\in \{1,2,\ldots,\hat{n}\}$).

A partition that maximize (\ref{approxQ}) in a given step must satisfy the geometric constraints $\textbf{R}_l \cdot \textbf{r}_i > 0$, $g_i=l$, and $\textbf{R}_l \cdot \textbf{R}_h < 0$ for all $l,h \in \{1,2,\ldots, \hat{n}\}$. Hence, if one constructs an approximation $\hat{Q}$ using $p$ eigenvectors, a network component can be split into at most $p+1$ groups in a given recursive step. The choice $p = 2$ allows either bipartitioning or tripartitioning in each recursive step. Reference \cite{Newmanspectral} discussed the case of general $p$ but reported results for recursive bipartitioning with $p = 1$.
   %Due to the geometry constraint, each permissible bipartition of the network can be specified as bisecting the 2D-plane of node vectors with a cut line passing through the original point. Therefore one only needs to find out the bipartition with highest modularity among all $N/2$ distinct permissible ones. This spectral bipartitioning scheme is also used recursively to look for the optimal partition with more than two clusters.
Reference \cite{tripart} implemented this spectral method with $p = 2$ and a choice of bipartitioning or tripartioning at each recursive step.

 % {\bf map: Huiyi: please check: did \cite{Newmanspectral} ever report numerical results for $p = 2$ bipartitioning in the examples?  I know the theory is there, but your phrasing was incorrect in several respects, and this part I don't remember exactly what is in that paper; also, I got rid of the details because I don't think they need to be discussed here}{\bf Huiyi: you are right.}

In Table \ref{table scinet}, we report diagnostics for partitions obtained by several algorithms (for $\gamma=1$).  For the recursive spectral bipartitioning and tripartitioning, we use {\sc Matlab} code that has been provided by the authors of Ref.~\cite{tripart}.  They informed us that this particular implementation was not optimized for speed,
%---but rather to illustrate the importance of going beyond bipartitioning steps in spectral methods---
so we expect it to be slow. One can create much faster implementations of the same spectral method. The utility of this method for the present comparison is that Ref.~\cite{tripart} includes a detailed discussion of its application to the network of network scientists.  Each partitioning step in this spectral scheme either bipartitions or tripartitions a group of nodes.  Moreover, as discussed in Ref.~\cite{tripart}, %one good (but locally optimal) partition contains $N_c=3$ communities that correspond well to known research collaborations among network scientists, 
a single step of the spectral tripartitioning is by itself interesting. 
Hence, we specify $\hat{n}=3$ for the Modularity MBO scheme as a comparison.

\begin{table}[h!]
\centering
\begin{tabular}{ |c | c | c |c|c|c |}
  \hline                        
   & $N_c$ &Q&NMI&Purity&Time (seconds)\\
 \hline 
  GenLouvain& 19&0.8500&1&1&0.5 s\\
   \hline
  Spectral Recursion& 39 & 0.8032&0.8935&0.9525&60 s \\
%  \hline
%  Spec. recur. (KL)&24&0.8427 & 0.9311&0.9446& 224s\\
  \hline 
   RMM &23&0.8344&0.9169&0.9367&0.8 s\\
   \hline
   Tripartition&3&0.5928&0.3993&0.8470&50 s\\
   \hline
   Modularity MBO&3&0.6165&0.5430&0.9974&0.4 s\\
   \hline
  \end{tabular}
  \caption{}\label{table scinet}
\end{table}

From Table \ref{table scinet}, we see that the Modularity MBO scheme with $\hat{n}=3$ gives a higher modularity than a single tripartition, and the former's NMI and purity are both significantly higher. When we do not specify the number of clusters, the RMM scheme achieves a higher modularity score and NMI than recursive bipartitioning/tripartitioning, though the former's purity is lower (which is not surprising due to its larger $N_c$). The RMM scheme and GenLouvain have similar run times.
%take similar amount of running time, while the spectral recursion requires much more.
For any of these methods, one can of course use subsequent post-processing, such as Kernighan-Lin node-swapping steps \cite{Newmanspectral, tripart,MasonNotice}, to find higher-modularity partitions. 

%{\bf map: and the comment above needs to be about our implementation, not spectral recursion in general}

%{\bf map: Huiyi, how did you apply the tripartitioning?  did you only apply one step or did you apply it recursively?  we need to discuss this, especially as you wrote incorrect statements about the method in your description of it; you need to specify very carefully what you did, and it looks to me like you applied it only a single time}

%%===========Discussion========%%%%%%%

\section{Conclusion and Discussion}\label{conc}

In summary, we have presented a novel perspective on the problem of modularity optimization by reformulating it as a minimization of an energy functional involving the total variation on a graph.  This provides an interesting bridge between the network science and compressive sensing communities, and it allows the use of techniques from compressive sensing and image processing to tackle modularity optimization.  In this paper, we have proposed MBO schemes that can handle large data at very low computational cost.  Our algorithms produce competitive results compared to existing methods, and they scale well in terms of speed for certain networks (such as the MNIST data). %large networks with a small number of communities. 
In our algorithms, after computing the eigenvectors of the graph Laplacian, the time complexity of each MBO iteration step is $O(N)$. 
%Compared with the leading algorithms for modularity optimization such the Genlouvain, the schemes we proposed produce competitive results. Especially for large scale networks with small number of clusters, our method achieves significant speed-up over the prior algorithms.

One major part of our schemes is to calculate the leading eigenvector-eigenvalue pairs, so one can benefit from the fast numerical Rayleigh-Chebyshev procedure in Ref.~\cite{RC} when dealing with large, sparse networks. Furthermore, for a given network (which is represented by a weighted adjacency matrix), one can reuse previously computed eigen-decompositions for different choices of initial functions, different values of $\hat{n}$, and different values of the resolution parameter $\gamma$.  This provides welcome flexibility, and it can be used to significantly reduce computation time because the MBO step is extremely fast, as each step is $O(N)$ and the number of iterations is empirically small.

%{\bf map: 'fast' convergence: meaning with few steps?  I don't think we ever reported the number of steps it takes in practice; shouldn't we report this?  also, I would like to clarify above that that is what we mean by 'fast convergence' above}

Importantly, our reformulation of modularity also provides the possibility to incorporate partial ground truth.  This can accomplished either by feeding the information into the initial function or by adding a fidelity term into the functional. (We only pursued the former approach in this paper.) It is not obvious how to incorporate partial ground truth using previous optimization methods.  This ability to use our method either for unsupervised or for semi-supervised clustering is a significant boon.

%{\bf [Huiyi: Mason, should we write something about Future directions? Multi-slice \cite{MasonSci}?]  }

%{\bf map: I wouldn't bother.  I would just leave that to discuss in the new paper whose calculations we have started. If the referees ask us, we could always add this later.  If we do this now, they might ask us to include some of these things here, and that is miles beyond the scope of this paper.}

%%%%%%%

\section*{Acknowledgements}

We thank Marya Bazzi, Yves van Gennip, Blake Hunter, Ekaterina Merkurjev, and Peter Mucha for useful discussions.  We also thank Peter Mucha for providing his spectral partitioning code.  We have included acknowledgements for data directly in the text and the associated references.

%{\bf map: do we need to bring up code usage here or is that already elsewhere? I know the funding stuff is already mentioned elsewhere; other code to mention? thanks for data for which we're not already citing a website}

%%%%%%%%%%======= Appendix ====%%%%%%%%%%%%%%%%%%%%%%%%%%%%%%%%%%%%

%\section*{Appendix}

\appendix
\section*{Appendix}
The notion of $\Gamma$-convergence of functionals is now commonly used for minimization problems. See Ref.~\cite{gamma-intro} for detailed introduction.  In this appendix, we briefly review the definition of $\Gamma$-convergence and then prove the claim that the graphical multi-phase Ginzburg-Landau functional $\Gamma$-converges to the graph TV. This proof is a straightforward extension of the work in Ref.~\cite{Yves} for the two-phase graph GL functional.

%{\bf map: note: the words being defined had the exact same font style as those around it; I didn't use bold because that was already being used; I have put 'tt' in there just to set the words apart, but I don't think that is a good choice; however, we have to set these words apart in *some* way}

\begin{definition}
Let $X$ be a metric space and let $\left\{F_n : X\rightarrow \mathbb{R}\cup\{\pm \infty\}\right\}^{\infty}_{n=1}$ be a sequence of functionals. The sequence $F_n$ {\tt $\Gamma$-converges} to the functional $F:X\rightarrow R\cup\{\pm \infty\}$ if, for all $f\in X$, the following lower and upper bound conditions hold:

\vspace{.2 in}
%\begin{minipage}{.8\textwidth}
\begin{description}
\item[(lower bound condition)] for every sequence $\{f_n\}_{n=1}^{\infty}$ such that $f_n\rightarrow f$, we have \[F(f)\leq \liminf_{n\rightarrow \infty}F_n(f_n)\,;\] 

\item[(upper bound condition)] there exists a sequence $\{f_n\}_{n=1}^{\infty}$ such that \[F(f)\geq \limsup_{n\rightarrow \infty}F_n(f_n)\,.\]
\end{description}
%\end{minipage}
\end{definition}

%\begin{definition}
%Let $X$ be a metric space. The sequence $\{F_n : X\rightarrow \mathbb{R}\cup\{\pm \infty\}\}^{\infty}_{n=1}$ is {\tt equi-coercive} if for every $C\in \mathbb{R}$, there exists a compact set $K_C \subset X$ such that $\{f\in X: F_n(f)\leq C\} \subset K_C$ for all $n$.
%\end{definition}

Reference \cite{MBO-multiGL} proposed the following multi-phase graph GL functional: 
\begin{align*}
	GL^{\mathrm{multi}}_{\epsilon}(\hat{f})=\frac{1}{2}\sum_{l=1}^{\hat{n}}\langle \hat{f}^{(l)},\textbf{L}\hat{f}^{(l)} \rangle +\frac{1}{\epsilon^2}\sum_{i=1}^N W_{\mathrm{multi}}(\hat{f}(n_i))
\end{align*}
where $\hat{f}: G\rightarrow \mathbb{R}^{\hat{n}}$ and $W_{\mathrm{multi}}(\hat{f}(n_i))=\prod_{l=1}^{\hat{n}}\|\hat{f}(n_i)-\vec{e}_l\|^2_{\ell_1}$.  See Sections \ref{meth} and \ref{section algorithm} for the definitions of all of the relevant graph notation. Let $X=\{\hat{f}\;| \; \hat{f}:G\rightarrow \mathbb{R}^{\hat{n}}\}$, $X^p=\{f \; | \;f:G \to V^{\hat{n}} \}\subset X$, and $F_{\epsilon}=GL^{\mathrm{multi}}_{\epsilon}$ for all $\epsilon>0$. Because $\hat{f}$ can be viewed as a matrix in $\mathbb{R}^{N\times \hat{n}}$, the metric for space $X$ can be defined naturally using the $\ell_2$ norm. 
\begin{theorem}\label{gamma convergence}
({\tt $\Gamma$-convergence}). The sequence $F_{\epsilon}$ $\Gamma$-converges to $F_0$ as $\epsilon \rightarrow 0^+$, where 
\begin{align*}
	F_0(\hat{f}):=
\begin{cases}
	|\hat{f}|_{TV}=\frac{1}{2}\sum_{i,j=1}^N w_{ij}\|\hat{f}(n_i)-\hat{f}(n_j)\|_{\ell_1}\,, \quad &\mathrm{if}~\hat{f}\in X^p\,,\\
	+\infty\,, \quad &\mathrm{otherwise}\,.
\end{cases}
\end{align*}
\end{theorem}
\begin{proof}
Consider the functional $W_{\epsilon}(f)=\frac{1}{\epsilon^2} \sum_{i=1}^N W_{\mathrm{multi}}(f(n_i))$ and 
\begin{align*}
	W_0(f):=
\begin{cases}
	0\,, \quad &\mathrm{if}~f\in X^p\,,\\
+\infty\,, \quad &\mathrm{otherwise}\,.
\end{cases}
\end{align*}

%{\bf map: possibly naive notation question: does $\mathbb{R}_+$ include 0 or not?  is there a standard convention for this?  when I was taking courses, such things seemed to vary across books and across lecturers}{\bf Huiyi: R_+ only includes positive real numbers. I think this notation is pretty clear.}

First, we show that $W_{\epsilon}$ $\Gamma$-converges to $W_0$ as $\epsilon \rightarrow 0^+$. Let $\{\epsilon_n\}_{n=1}^{\infty}\subset (0,\infty)$ be a sequence such that $\epsilon_n \rightarrow0$ as $n\rightarrow \infty$. For the lower bound condition, suppose that a sequence $\{f_n\}_{n=1}^{\infty}$ satisfies $f_n\rightarrow f$ as $n\rightarrow \infty$. If $f\in X^p$, then it follows that $W_0(f)=0\leq \liminf_{n\rightarrow \infty}W_{\epsilon_n}(f_n)$ because $W_{\epsilon}\geq 0$. If $f$ does not belong to $X^p$, then  there exists $i\in \{1,2,\ldots,N\}$ such that $f(n_i)\not\in V^{\hat{n}}$ and $f_n(n_i) \rightarrow f(n_i)$. Therefore, $\liminf_{n\rightarrow \infty}W_{\epsilon_n}(f_n)=+\infty  \geq W_0(f) = +\infty$. For the upper bound condition, assume that $f\in X^p$ and $f_n=f$ for all $n$. It then follows that $W_0(f)=0\geq  \limsup_{n\rightarrow \infty}W_{\epsilon_n}(f_n)=0$. Thus, $W_{\epsilon}$ $\Gamma$-converges to $W_0$.

Because $Z(f):=\frac{1}{2}\sum_{l=1}^{\hat{n}}\langle f^{(l)},\textbf{L}f^{(l)} \rangle$ is continuous on the metric space $X$,  it is straightforward to check that the functional $F_{\epsilon_n}=Z+W_{\epsilon_n}$
satisfies the lower and upper bound condition and therefore $\Gamma$-converges to $Z+W_0$.
 %is a continuous perturbation of $W_{\epsilon_n}$ and $\Gamma$-converges to $Z+W_0$. 
 
 Finally, note that $Z(f)=|f|_{TV}$ for all $f \in X^p$. Therefore, $Z+W_0=F_0$ and one can conclude that $F_{\epsilon_n}$ $\Gamma$-converges to $F_0$ for any sequence $\epsilon_n \rightarrow0^+$.
\end{proof}

%%%%%%%%%%======= Reference====%%%%%%%%%%%%%%%%%%%%%%%%%%%%%%%%%%%%

\end{document}